\documentclass[11pt]{article}
\usepackage{amsmath,amssymb,amsthm,mathtools}
\usepackage{geometry}
\usepackage{amsthm}
\usepackage{physics}
\usepackage{verbatim}
\usepackage{amsmath, amssymb, amsthm}
\usepackage{physics}         
\usepackage{mathtools}       
\usepackage{hyperref} 
\usepackage{orcidlink}
\usepackage{cite}            
\usepackage{bm}
\usepackage{color,soul} 
\usepackage{xcolor}
\usepackage[most]{tcolorbox}
\usepackage{cancel}

\newtcolorbox{important}[1][]{
  colback = gray!6,            
  colframe = gray!40,            
  boxrule = 0.9pt,                
  arc = 3pt,  
  left=10pt, right=10pt,
  top=8pt, bottom=8pt,
  enhanced,
  #1
}

\newtheorem{definition}{Definition}[section]

\newtheorem{theorem}[definition]{Theorem}
\newtheorem{lemma}[definition]{Lemma}
\newtheorem{remark}[definition]{Remark}

\newtheorem{corollary}[definition]{Corollary}
\newtheorem*{theorem*}{Theorem}

\author{A. Della Corte$^1$\footnote{Corresponding author; email to: alessandro.dellacorte@unicam.it.}\orcidlink{0000-0002-1782-0270}, M. Farotti$^2$\orcidlink{0009-0001-5000-2827}, L. Guglielmi$^{1,3}$\orcidlink{0009-0006-1624-049X}}
\date{
{\normalsize
$^1\;$\footnotesize{\textit{School of Science and Technology, University of Camerino,\\
via Madonna delle Carceri, 9, Camerino, I-62032, Italy}}\\
\vspace{3mm}
$^2\;$\footnotesize{\textit{Doctoral School in Computer Science and Mathematics, University of Camerino,\\ Via Madonna delle Carceri, 9, Camerino, 62032, Italy}}\\
\vspace{3mm}
$^3\;$\footnotesize{\textit{Istituto Nazionale di Fisica Nucleare, Sezione di Perugia,\\
via A.~Pascoli, I-06123 Perugia, Italy}}
}}
\title{A no-go theorem for irreversibility in\\ arbitrary realizations of the collapse dynamics}

\begin{document}
\maketitle

\begin{abstract}
\noindent We study finite dimensional quantum systems with arbitrary collapse events, establishing a structural no-go for operational irreversibility along arbitrary realizations of the collapse dynamics. More precisely, we prove that, for every choice of a physically admissible trajectory (i.e., collapse outcomes having nonzero Born weight) assigned to each state, there exists a nonempty topologically closed subset of the projective state space within which any two states can be connected with arbitrarily fine Fubini-Study precision and arbitrarily small integrated energetic cost. This shows that the preservation of information along  observed realizations of outcomes guarantees islands of quasi-reversibility, while genuine irreversibility requires additional ingredients such as non-compactness or information erasure.

\noindent KEYWORDS: Quantum collapse dynamics; Quasi-reversibility; Chain-recurrence; Information non-erasure.

\noindent MSC2020: 81P15; 81P45; 37B20; 03E10.

\end{abstract}

\tableofcontents

\section{Introduction}
\subsection{Reversibility in unitary versus collapse dynamics}
\noindent In unitary quantum dynamics of finite-dimensional systems, recurrence phenomena, either in the measurable or topological sense, arise from the continuity of unitary evolution. Indeed, since the Schr\"odinger dynamics is continuous (indeed isometric), in the compact case the classical recurrence theorems by Poincar\'e~\cite{poincare1890probleme,bocchieri1957quantum} and Birkhoff~\cite{birkoff1912quelques} ensure, respectively, that almost every orbit admits arbitrarily close returns in state space, and that quasi–periodic orbits exist. 
Therefore, within purely unitary evolution of compact systems, the dynamics admits no
intrinsic arrow of time. 

\noindent When wavefunction collapse is included in the picture, this structure breaks. For the discrete–time setting, a natural analogue of the previously described unitary evolution is obtained by following the collapse dynamics of states in an open system, an approach that has been introduced for the first time (within quantum optics) in \cite{dalibard1992wave}, and subsequently elaborated in many other works (for instance \cite{carmichael1993open} and \cite{wiseman2009quantum}). Consider a finite-dimensional space $\mathbb{P}(\mathcal{H})$ of pure states, a single observable $A$ with finite spectrum, and a particular state $\psi$. At each time step, either no collapse occurs (so one applies a unitary Schr\"{o}dinger evolution $\psi\mapsto U\psi$, or a more general bounded linear operator generated by an effective Hamiltonian) or a projective collapse $\{P_i\}$ occurs. We indicate a particular future outcome itinerary that is realized by $\psi$ by a sequence $\omega=\{i_0,i_1,\dots\}$, where $i_k\in\{0,\dots, N\}$ are symbols representing the possible collapse outcomes as well as the ``blank", no-collapse case. We call $\omega$ the \textit{selector} associated, in that particular forward history, to $\psi$. Any forward itinerary has to be of course admissible from a physical point of view, in the sense that each selected outcome must have a nonzero Born weight for the state immediately preceding that step. 
The forward evolution of the open system, that is the set of states observed in a particular forward-history after having prepared the system in a given state, is called a \textit{trajectory} (see Def-\ref{traj_}): 
$$
\{T^{(k)}(\psi)\}_{k=1,2,\dots},
$$
where the collapse outcome at time step $k$ is given by the symbol $i_k$ of the selector $\omega$ associated to the state $\psi$.

Notice that every trajectory associates to a state a particular forward-history, but we must observe explicitly that compatibility of distinct trajectories is not enforced, in the sense that, if the state $\psi$ arrives, along a particular trajectory, to the state $T^{(k)}(\psi)=\phi$ at time $k$, then there is no reason to expect that the trajectory \textit{starting} at $\phi$ coincides with the extension, beyond $\phi$, of the trajectory starting at $\psi$; for instance, there is no reason to expect that $T^{(k+1)}(\psi)=T^{(1)}(\phi)$. The two trajectories represent simply two possible forward-histories starting respectively from $\psi$ and $\phi$, and if we accept intrinsic unpredictability of collapse outcomes, they must be physically admissible but are otherwise arbitrary.

In general, trajectories are many–to–one and highly discontinuous, because there is no justification in assuming that close states will evolve into close collapsed states. Arbitrarily close states are allowed, in general, to collapse to eigenstates belonging to distinct eigenspaces of the chosen observable; this is not at all surprising
if one considers that even the \textit{same} pre–collapse state may be sent to different
post–collapse eigenstates at distinct collapse events. Moreover, the eigenspaces of the observable are in general not preserved by $U$ (i.e.\ $[U,P_i]\ne 0$), which means that, between collapses, the free evolution typically carries states out of these subspaces. Consequently, the hypotheses behind the Poincar\'e/Birkhoff recurrence theorems are violated on the observed system.\footnote{Of course if one re–includes apparatus and environment as a closed finite-dimensional composite, the global pure state again evolves unitarily on a compact projective space, and recurrence can be recovered.}

\noindent We call a \textit{sample realization of the collapse dynamics} any assignment of a trajectory to each state (see Def.\ref{traj_}). We want to establish a result that holds for any possible sample realization of the collapse dynamics. We want to do this while taking the selector in its most literal observational sense:
\begin{enumerate}
    \item each selected outcome has to have positive Born weight at that step;
    \item treating collapse as intrinsically unpredictable at the single-event level\footnote{As it is assumed in the standard von Neumann–Dirac formulation. Alternative interpretations, such as Bohmian mechanics or the Everett (many-worlds) interpretation, eliminate collapse entirely; the present work follows the conventional projective measurement framework.}, there is no basis to assume that an entire open neighborhood is shielded from any outcome with nonzero Born weight, which means that we have, in general, topological accessibility of every outcome label arbitrarily close to every state.
\end{enumerate}
The seemingly intractable irregularity that such state–wise selectors permit is usually sidestepped, for instance by averaging over outcomes or retaining stochasticity, where continuity and stability are recovered in measure or in expectation rather than pointwise for a deterministic selector–induced map \cite{carmichael1993open,belavkin1992quantum,barchielli2009quantum,attal2006repeated, wiseman2009quantum}. Our approach is instead to explore what can be said, from the point of view of recurrence, assuming only the physics-dictated property 1. and, thus, allowing 2.; the latter will be formalized through the density property \eqref{densityproperty}. In particular, we allow that the induced evolution be everywhere discontinuous. 

\subsection{The result}

Landauer’s principle (see \cite{landauer1961irreversibility}) bounds the minimal heat cost of logically irreversible operations (such as erasure), and conversely permits, in principle, arbitrarily low dissipation for logically reversible transformations in the quasi–static limit. The choice of a sample realization for the collapse dynamics can be interpreted as realizing the extreme, no-information-erasure limit of Landauer's principle, in the sense highlighted by C.H. Bennett's resolution of the Maxwell demon paradox \cite{bennett1982thermodynamics,bennett2003notes,hanson2018landauer}. In particular, we share the following position, expressed by Bennett (from \cite{bennett1982thermodynamics}): 

\noindent \emph{ ``[...] the essential irreversible step, which prevents the demon from breaking the second law, is not the making of a measurement (which in principle can be done reversibly) but rather the logically irreversible act of erasing the record of one measurement to make room for the next."} 

\noindent Picking a sample realization of the collapse dynamics, indeed, entails that we do not erase (or coarse-grain) outcome information: the entire outcome sequence is available as a fixed (though arbitrary) record.

As a physical principle, however, Landauer's is permissive rather than generative: it does not assert that a given realization of the collapse dynamics actually exhibits reversibility, nor does it identify when it must occur. In the present work we show that, despite the absence of regularity assumptions on the realization of the collapse dynamics, there exists a topologically closed, internally strong chain–transitive (see Def. \ref{intern_chain_tr_open}) subset of projective state space along which any two states can be connected with arbitrarily small Fubini–Study error and arbitrarily small integrated energetic cost (in the sense of Def.~\ref{def_cost_}). While Landauer’s principle says that, without erasure, quasi–reversibility is not thermodynamically forbidden, we prove that in finite dimension, an island of operational reversibility is topologically enforced.

\bigskip

\noindent Let us make this quantitative (for a fully formal treatment, see Sections \ref{DTD} and \ref{sec_4}).
We consider a finite forward trajectory
$$(T^{(0)}(\psi),\dots T^{(n)}(\psi) )=: \gamma$$
\textit{allowed} if there is a self-adjoint perturbation of the Hamiltonian of the system under which the state $\psi$ evolves along $\gamma$. Given a finite forward trajectory, we denote by
   $$\mathcal{E}_D(\gamma) \,\ge 0$$
its energetic cost, that is the minimum integrated energetic cost of a self-adjoint perturbation $\Delta H$ of the effective Hamiltonian $H_{\text{eff}}$ such that, under the Hamiltonian $H'=H_{\text{eff}}+\Delta H$, $\psi$ evolves along $\gamma$ (this will be made precise in Def.~\ref{def_energ_cost}). 

Given two states $\psi,\phi \in \mathbb{P}(\mathcal{H})$, we can then define the
\emph{minimal forward cost} of going from $\psi$ to $\phi$ in a given sample realization of the collapse dynamics:
$$
   d_D\big(\psi \to \phi\big)
   \;:=\;
   \inf\Big\{ \mathcal{E}_D(\eta) \;:\;
        \eta \text{ is an allowed finite forward itinerary
        from } \psi \text{ to } \phi \Big\}.
$$

\begin{definition}\label{arrow}
For a sample realization of the collapse dynamics, we say that a subset of states $B\subseteq \mathbb{P}(\mathcal{H})$ exhibits an
\emph{operational arrow of time} if there exist states
$\psi,\phi \in B$ such that
$$
   d_D\big(\psi \to \phi\big)
   \;<\;
   d_D\big(\phi \to \psi\big),
$$
i.e.\ the minimal energetic effort required to perturb the Hamiltonian so as to carry the
system from $\psi$ to $\phi$ in forward time is strictly smaller than the
minimal energetic effort required to realize the reverse change from
$\phi$ to $\psi$.

Conversely, we say that
the operational arrow of time does not arise (or that $B$ is
\emph{operationally reversible}) if
$$
   d_D\big(\psi \to \phi\big)
   \;=\;
   d_D\big(\phi \to \psi\big)
   \;=\; 0
   \qquad
   \text{for all } \psi,\phi \in B.
$$
\end{definition}

The previous definition is most meaningful when arbitrarily low-cost forward trajectories between states in $B$ can be realized entirely within $B$. This condition will be formalized through the concept of internal chain transitivity (see Def. \ref{intern_chain_tr_open}).

We can summarize our main result informally as follows. 

\begin{important}
\textbf{Main result (informal)} 
Given any finite-dimensional quantum system and any realization of the collapse dynamics, there exists a nonempty, topologically closed subset $S\subseteq\mathbb{P}(\mathcal{H})$ that is operationally reversible, meaning that, in $S$, quasi-returns with arbitrarily small Fubini-Study error and arbitrarily small energetic cost are allowed between any pair of states and passing only through states contained in $S$.  
\end{important}
Genuine irreversibility, thus, requires additional ingredients: information loss/coarse–graining, non-compact limits, or coupling to reservoirs.
This establishes a structural no-go: in finite dimension, collapse alone cannot generate an operational arrow of time along any sample realization of the collapse dynamics. 

Let us conclude by recalling that irreversibility and the arrow
of time for repeated quantum measurements are sometimes meant in a different sense. For instance in \cite{benoist2018entropy}, the focus is not on what outcome sequences are possible in the forward and backward directions with vanishing energetic cost, but how their probabilities compare. 

\subsection{Where a naive diagonal grid argument fails}
\noindent It may look tempting to try to get a (naive) heuristic argument proving our main result as follows: fix $\epsilon>0$, cover $\mathbb{P}(\mathcal{H})$ by finitely many $d_{FS}$–balls (“$\epsilon$–cells”), and note that, along any long enough portion of a trajectory, some cell must be visited twice; this yields an $\epsilon$–loop. Letting $\epsilon$ go to 0 seems to promise finer and finer quasi-loops. However, this argument is intrinsically \emph{scale–wise}: as $\epsilon$ shrinks, the repeating cell and the base point inside it drift with $\epsilon$. Of course, by compactness the chosen cells admit a convergent subsequence, but in the absence of continuity the trajectory at the limit point is decoupled from the finite $\epsilon$–loops, and in general it bears no relation to the trajectories seen along the approximants. Consequently, the grid argument does not deliver a chain–recurrent point, let alone a topologically closed, internally strong chain–transitive subset (which, as we will see, is the desired result); nothing in this argument enforces \emph{nesting across scales} of the pseudo–orbits.

\noindent To make this remark more precise, fix a mesh $\epsilon_k=2^{-k}$ and choose refining covers $\mathcal{U}_{k+1}\prec\mathcal{U}_k$ by $d_{FS}$–balls of radius $<\epsilon_k$. 
At scale $k$, scan a long forward portion of a trajectory and pick the first cell $C_k\in\mathcal{U}_k$ that is visited twice; let $\Gamma_k$ be the loop between those visits. 
This yields strong $\epsilon_k$–loops, but the selected cells need not nest: typically $C_{k+1}\nsubseteq C_k$. 
Even after extracting a nested subsequence $C_{k_r}\searrow\{x_\ast\}$, the discontinuity of the trajectories decouples the local dynamics at $x_\ast$ from the images traced by the $\Gamma_{k_r}$, so neither $x_\ast\,\mathcal{SC}_{d_{FS}}\,x_\ast$ (where $\mathcal{SC}_{d_{FS}}$ is the strong chain-relation defined in Def.~\ref{strong_ch_}) nor an internally $\mathcal{SC}_{d_{FS}}$–transitive set is guaranteed. Summarizing, the grid-refining technique only shows that: 
$$\text{For every $\epsilon$, there is a cell revisited with $O(\epsilon)$ accuracy and energetic cost,}$$ 
while we search for the stronger, dynamically relevant property: 
$$\text{There is a state such that, for every $\epsilon$, it is revisited with $O(\epsilon)$ accuracy and energetic cost.}$$

\subsection{Constructive Result over Non-Constructive Existence}

At the opposite extreme of the naive attempt  described above, one could obtain an abstract pure existence result on the existence of a chain-transitive subset by using Zorn's lemma on topologically closed, trajectory-invariant subsets ordered by inclusion. While this method does provide an existence result for a closed and chain-transitive subset of the state space, it acts as a non-constructive ``black box" that obscures the very operationality we seek to characterize. 

Our aim, instead, is to give a \textit{constructive} procedure, which works well with the operational character of the employed notion of reversibility. Therefore, not relying on the Axiom of Choice, we implement (in Section~\ref{sec_proof}) a transfinite argument that enforces both spatial nesting and dynamical coherence across scales. 
This procedure yields a strong chain–recurrent point, and even a topologically closed subset that is internally $\mathcal{SC}_{d_{FS}}$–transitive (see Sections \ref{sec_prel} and \ref{DTD} for the relevant definitions). 

Even more importantly, the two basic operations considered in the procedure, that is passing from the $k$-th to the $(k+1)$-th state in a trajectory, and passing to an accumulation point of a trajectory, have finite-precision analogues. The transfinite induction serves thus as a rigorous proxy for an arbitrary-precision search protocol, which Zorn's lemma (or in general arguments intersecting a descending chain of suitable subsets, even if relying on weaker forms of choice) cannot provide.

\bigskip

The paper is organized as follows: in Section \ref{sec_prel} we introduce basic concepts and notation for the free evolution of the system and the abstract topological dynamics concepts that will be used, in a suitably adapted form, later; in Section \ref{DTD} we describe the selector of the collapse dynamics and the induced discrete-time evolution; in Section \ref{sec_4} we study reversibility for arbitrary realizations of the collapse dynamics; Section~\ref{sec_proof} is devoted to the proof of Theorem \ref{th_rec_}; Section \ref{sec_concl} draws some conclusions.

\section{Preliminaries}\label{sec_prel}

\subsection{Notation and definitions for the free evolution}

Let $\mathcal{H} \cong \mathbb{C}^{n+1}$ be a finite-dimensional Hilbert space (with $n \ge 0$). 
Note that $\dim_{\mathbb{C}}\mathcal{H} = n+1$, so the associated complex projective space $\mathbb{P}(\mathcal{H}) \cong \mathbb{CP}^{n}$ has complex dimension $n$. 
Let $\mathbb{P}(\mathcal{H})$ be the complex projective space of one-dimensional subspaces (pure states), equipped with the Fubini--Study metric $d_{FS}$ (see Def.~\ref{def_FSmetric} below), which makes $\mathbb{P}(\mathcal{H})$ a compact metric space. 

We recall the definition of the Fubini-Study metric (for a more detailed discussion, see~\cite{bengtsson2017geometry}). 
\begin{definition}[Fubini-Study metric]\label{def_FSmetric}
Let $\mathbb{CP}^n$ be the complex projective space of complex
dimension $n$ (note that $\mathbb{P}(\mathbb{C}^{n+1})$ is an equally standard notation). A point in $\mathbb{CP}^n$ can be represented by homogeneous
coordinates $[z_0 : z_1 : \ldots : z_n]$, where
$(z_0, z_1, \ldots , z_n) \in \mathbb{C}^{n+1}\setminus\{0\}$ and
$$
[z_0 : z_1 : \ldots : z_n] = [\lambda z_0 : \lambda z_1 : \ldots : \lambda z_n]
\quad \text{for all } \lambda \in \mathbb{C}\setminus\{0\}.
$$
Let $\{\ket{e_k}\}_{k=0}^{n}$ be an orthonormal basis for $\mathcal{H}$, and
consider two state vectors
$$
\ket{\phi} = \sum_{k=0}^n z_k \ket{e_k},
\qquad
\ket{\psi} = \sum_{k=0}^n w_k \ket{e_k},
$$
representing the points $[\phi], [\psi] \in \mathbb{CP}^n$,
respectively.  
Then the \emph{Fubini--Study distance} between $[\phi]$ and $[\psi]$ is defined by
\begin{equation}\label{eq_FS_metric}
    d_{FS}([\phi],[\psi]) 
    = \arccos
    \left(
        \sqrt{\frac{|\braket{\phi}{\psi}|^2}
        {\braket{\phi}{\phi}\,\braket{\psi}{\psi}}}
      \right).
\end{equation}
Equivalently, the Fubini--Study metric is the Kähler metric whose associated
$(1,1)$-form is
$$
\omega_{FS} = \frac{i}{2}\,\partial \bar{\partial} \log \|z\|^2,
\qquad
\text{where } \|z\|^2 = \sum_{j=0}^{n} |z_j|^2.
$$
\end{definition}

We indicate by $\pi\colon \mathcal H\setminus\{0\}\to \mathbb P(\mathcal H)$ the canonical projection on the projective space, defined by
$$
\pi(u):=[u].
$$
Let $\{P_j\}_{j\in\Lambda}$ be a finite collection of orthogonal projections on $\mathcal{H}$, indexed by a finite alphabet 
$$
\Lambda=\{0,1,\dots,m\}.
$$ 
We do \emph{not} assume any property for the projections $P_j$ beyond linearity, orthogonality, and idempotence. We interpret these projections $\{P_j\}_{1\leq j\leq m}$ as the spectral projections of a single self-adjoint observable $A$ on $\mathcal{H}$: there exist distinct eigenvalues $\{\lambda_j\}_{j=1}^{m}$ such that
$$
A=\sum_{j=1}^{m}\lambda_j P_j,\qquad P_jP_k=\delta_{jk}P_j,\qquad \sum_{j=1}^{m}P_j=I,
$$
where $P_j$ projects onto the (possibly degenerate) eigenspace $E_{\lambda_j}=\operatorname{im}P_j$.
The index $0$ is reserved for the ``blank'' (no-collapse) channel; therefore, we set $P_0:=I$, which is not part of the spectral resolution of $A$. 

The continuous, no-jump evolution is generated by the effective non-self-adjoint Hamiltonian
$$
H_{\text{eff}}=H-\frac i 2 \sum_{j=1}^m L_j^\dagger L_j 
$$ 
describing the conditional evolution of the system between collapse events, where $L_j$ are the Lindbladian bounded operators (see, for instance, \cite{barchielli2009quantum}) on the Hilbert space and $H$ is the time-independent system Hamiltonian describing the unitary evolution.

Since we consider a single observable $A$ on $\mathcal{H}$, we assume $L_j=\sqrt{\gamma_j} P_j$ with $\gamma_j>0$ for every $j=1,\ldots,m$ (see, for instance, \cite{hebenstreit2015real}).
Set 
$$\Gamma:=\sum_{j=1}^m L_j^\dagger L_j=\sum_{j=1}^m \gamma_j P_j \qquad \text{and} \qquad \gamma_{\text{max}}:=\max_{j=1,\ldots,m} \gamma_j.
$$
\begin{remark}\label{rem_gamma_norm}
    Notice that $0\le \Gamma\le \gamma_{\text{max}} I$.
    Therefore $\|\Gamma\|\le \gamma_{\text{max}}$.
\end{remark}
The linear effective operator $U\in GL(\mathcal{H})$ represents the free continuous, no-jump evolution. 
Let $U(t_1,t_0)$ denote the evolution through $U$ from time $t_0$ to time $t_1$ defined as
$$
U(t_1,t_0)=e^{-iH_{\text{eff}} (t_1-t_0)},
$$
where we use the normal units: $\hbar=1.$
The operator $U$ satisfies 
$$
U(t_2,t_1)U(t_1,t_0)=U(t_2,t_0) \text{ for $t_0<t_1<t_2$}
$$
$$
U(t,t)=I,
$$
where $I$ is the identity operator.
In the following, when we have $t_0=0$, we use the notations $U(t)\coloneqq U(t,0)$ and also set $U\coloneqq U(1)=U(1,0)$.
Since $H_{\text{eff}}$ is not self-adjoint $U (t_1,t_0)^\dagger U (t_1,t_0)\neq I$ in general; nonetheless $U (t_1,t_0)$ is always invertible with $U(t_1,t_0)^{-1}=U(t_0,t_1)$.

\begin{remark}\label{rem_cond_numb}
 Notice that $\|U(t,t_0)\|\le 1$ for $t\ge t_0$. 
 Indeed, for $u\in\mathcal{H}$, setting $u_t:=U(t,t_0)u$, we have $\partial_t u_t=-iH_{\text{eff}}\,u_t$ with $u_{t_0}=u$. 
 Hence,
 \begin{align*}\label{ineq_der_norm}
     \partial_t \|u_t \|^2 &= i\langle u_t | H_{\text{eff}}^\dagger \,u_t \rangle- i\langle u_t | H_{\text{eff}} \,u_t\rangle
     =i\langle u_t | (H_{\text{eff}}^\dagger-H_{\text{eff}}) \,u_t\rangle=-\langle u_t| \Gamma\, u_t\rangle \le 0
 \end{align*}
 where in the last equality we use $H_{\text{eff}}^\dagger-H_{\text{eff}}=H+\frac i 2 \Gamma -\left( H -\frac i 2 \Gamma \right)=i\Gamma$.
 Therefore, $\| u_t\|\le \|u_{t_0}\|=\|u\|$ from which follows $\|U(t,t_0)\| \le 1$.
\end{remark}

\begin{remark}\label{rem_lips}
    For $t_1\ge t_0$, the map $[u]\mapsto [U(t_1,t_0)u]$, being a $C^\infty$ map defined on a compact manifold, is globally Lipschitz.
\end{remark}

\begin{definition}
    For $t_1\ge t_0$, we define the evolution of $u\in\mathcal{H}$ under the normalized effective operator $\widehat{U}(t_1,t_0)$ as
    $$
    \widehat{U}(t_1,t_0)\,u= \frac{U(t_1,t_0)\,u}{\| U(t_1,t_0)\,u\|}.
    $$
\end{definition}

A single collapse event with outcome index
$j\in\Lambda$ acts (where defined) by
$$
[u]\mapsto f_j([u]) \coloneqq \frac{P_j u}{\lVert P_j u\rVert },
$$
where $u\in\mathcal{H}$ is any unit vector that represents the point
$[u]\in \mathbb{P}(\mathcal{H})$, and the expression is defined when $\lVert P_j u\rVert\ne 0$.

\subsection{Notation and definitions for Topological Dynamics formalism}

In this Section we provide some topological dynamical background concerning the concept of strong chain relation. In the following, we will need to adapt these concepts to the case in which we have a sample realization of the open dynamics instead of a dynamical system, that is an $\mathbb{N}$-action.

We say that $(X,f)$ is a topological dynamical system if $X$ is a topological space with a metric $d:X^2\to\mathbb{R}_0^+$  compatible with the topology on $X$ and $f:X\to X$ a map.
By $f^n(x)$ with $x\in X$ and $n\in\mathbb{N}$ we mean the $n$-times composition $f\circ f\circ \ldots \circ f(x)$. We say that $(X,f)$ is a \emph{compact dynamical system} if we require, in addition, that $X$ is a compact metric space.

Let us recall the notion of $\epsilon$-{\em chain} or $\epsilon$-{\em pseudo-orbit} \cite[p. 48]{kurka2003topological}.

\begin{definition}
Given two points $x,y\in X$ and $\epsilon>0$, an $\epsilon$-\emph{chain} (also called $\epsilon$-pseudo-orbit) from $x$ to $y$ is a finite set of points $x_0,x_1,\ldots,x_n$ in $X$, with $n\ge 1$, such that
\begin{enumerate}
\item[i)] $x_0=x$ and $x_n=y$,
\item[ii)] $d(f(x_i),x_{i+1})<\epsilon$ for every $i=0,1,\ldots n-1$.
\end{enumerate}
The {\it chain relation} $\mathcal{C}\subseteq X^2$ is the binary relation defined as follows: given $x,y\in X$, 
$$
x\,\mathcal{C}\, y\iff \forall \epsilon>0\text{ there exists an $\epsilon$-chain from $x$ to $y$}.
$$
We say that $x\in X$ is a \emph{chain-recurrent point} of $(X,f)$ if $x\,\mathcal{C}\,x$.
\end{definition}

\begin{definition}
    Let $\mathcal{A}\subseteq X^2$ be a binary relation on $X$. For $S\subseteq X$ and $y\in X$, we write $S\,\mathcal{A}\, y$ if $x\,\mathcal{A}\, y$ for every $x\in S$. 
\end{definition}

Recall the concept of strong chain-recurrence originally introduced by Easton \cite{easton2006chain}.

\begin{definition}
Given two points $x,y\in X$ and $\epsilon>0$, we say that a finite sequence of points $x_0,x_1,\ldots,x_n $ of $X$, with $n\in\mathbb{N}$, is a \emph{strong $(\epsilon,d)$-chain} from
$x$ to $y$ if 
\begin{enumerate}
    \item[i)] $x_0 = x$ and $x_n = y$,
    \item [ii)] $ \sum_{i=0}^{n-1} d(f(x_i),x_{i+1}) <\epsilon$.
\end{enumerate}
The \emph{strong chain relation} $\mathcal{SC}_d\subseteq X^2$ is the binary relation defined as follows: given $x, y \in X$,
$$
x\, \mathcal{SC}_d\, y \iff \forall \epsilon>0\text{ there exists a strong $(\epsilon,d)$-chain from $x$ to $y$}.
$$
We say that $x\in X$ is a \emph{strong chain-recurrent point} of $(X,f)$ if $x\,\mathcal{SC}_d\,x$.
\end{definition}

\begin{remark}
    Notice that, in general, the relation $\mathcal{SC}_d$ depends on the metric $d$, while the relation $\mathcal{C}$, when the space $X$ is compact, can be equivalently defined in purely topological terms (see, for instance, \cite[p.~47]{kurka2003topological}).
\end{remark}

\begin{definition}\label{def_strong_trans}
    We say that $S\subseteq X$ is \emph{strong chain-transitive} (\emph{$\mathcal{SC}_d$-transitive}) if for every $x,y\in S$ we have $x\,\mathcal{SC}_d\, y$ and $y\,\mathcal{SC}_d\, x$. 
\end{definition}

The following is a strengthening of $\mathcal{SC}_d$-transitivity for subsets, and ensures that a subset $S\subseteq X$ is $\mathcal{SC}_d$-transitive also as a subsystem, because no points outside $S$ are needed to build the chains (see, for instance, \cite{hirsch2001chain}).

\begin{definition}\label{def_int_s_chain}
    We say that $S\subseteq X$ is \emph{internally $\mathcal{SC}_d$-transitive} if for every $x,y\in S$ and for every $\epsilon>0$ there exists a strong $(\epsilon,d)$-chain from $x$ to $y$ whose points belong to $S$.
\end{definition}

\begin{definition}
The \emph{derived set} $S'$ of a subset $S\subseteq X$ is the set of the limit points $x$ of $S$, that is, the set of points $x\in X$ such that for every neighborhood $U$ of $x$ we have $S\cap (U\setminus\{x\})\neq \emptyset$.
Moreover, we denote by $\overline{S}$ the topological closure of $S$, that is the smallest closed set containing $S$. In particular we have $\overline{S}=S\cup S'$.
\end{definition}

\vspace{0.5cm}

Throughout the paper, we use the standard Landau notation $O(\epsilon^k)$ to denote terms whose norm is bounded by $C\,|\epsilon|^k$ for some constant $C>0$ independent of $\epsilon$.

\section{Discrete-time dynamics}
\label{DTD}

\begin{definition}
For each $j\in\Lambda$, set
$$
D_j := \{ [u]\in \mathbb{P}(\mathcal{H}): \lVert P_j u\rVert\ne 0 \}.
$$
Equivalently, $D_j$ is the set of pure states for which the outcome $j$ has strictly positive Born probability.
Notice that $D_0=\mathbb{P}(\mathcal{H})$.
\end{definition}

For each $j\in\Lambda$ the map $f_j:D_j\to \mathbb{P}(\mathcal{H})$ defined by $f_j([u])=\dfrac{P_j u}{\lVert P_j u\rVert}$ is continuous on $D_j$. 

Let $\Omega=\Lambda^{\mathbb{N}_0}=\{(w_n)_{n\in\mathbb{N}_0} : w_n\in \Lambda \text{ for all } n\in\mathbb{N}_0\}$ be the one-sided full shift with the
product topology (equivalently the metric $d_\Omega(\omega,\omega')=2^{-N}$
where $N$ is the largest positive integer such that $\omega_n=\omega'_n$ for
all $n<N$).  Denote the (continuous) left shift by
$\sigma:\Omega\to\Omega$.  We can equip the product $S:=\Omega\times \mathbb{P}(\mathcal{H})$ with
the product metric
$$
d_S\bigl((\omega,[u]),(\omega',[u'])\bigr) = d_\Omega(\omega,\omega') + d_{FS}([u],[u']).
$$
Set
$$
\quad \mathcal{S} = \{ (\omega,[u]) \in S : [Uu] \in D_{\omega_0} \}.
$$

\begin{definition}
Define the skew-product map $F:\mathcal{S}\to S$ by
$$
\quad F(\omega,[u]) = (\sigma\omega, f_{\omega_0}([Uu])).
$$
Here $\omega_0$ denotes the zeroth coordinate of $\omega$.
\end{definition}

\noindent Note that $F$ is continuous because each fiber map $f_j$ is continuous on $D_j$ and $\sigma$ is continuous.

\begin{definition}
A \emph{selector} is any function
$$
D:\mathbb{P}(\mathcal{H})\to\Omega,\qquad [u]\mapsto D([u])=(D([u])_n)_{n\in\mathbb{N}_0}.
$$
\end{definition}

It is clear that not every selector makes physical sense. We have to ensure that collapse events at each point produce observable outcomes only having positive probability.
This motivates the following request.

\begin{definition}
A selector $D:\mathbb{P}(\mathcal{H})\to\Omega$ is called \emph{admissible} if it always selects a collapse outcome that has nonzero Born weight. More precisely, we require that, for every $k\geq 1$,
$$[U (F^{k-1})_2(D([Uu]),[u])] \in D_{D([Uu])_{k-1}},$$
where $(F^k)_2$ denotes the second component  of the $k$-th iteration of the skew-product map $F$. 
\end{definition}
We denote by $\mathfrak{S}$ the set of all admissible selectors.
\begin{definition}
\label{traj_}
For $U$ a linear operator and $D$ an admissible selector, we call \emph{trajectory} of a state $[u]$ the sequence of states given by:
 $$
 T^{(k)}(U,D)([u]):= (F^k)_2(D([Uu]),[u]), \quad k\in\mathbb{N}.
 $$
We define a \emph{sample realization of the collapse dynamics}, denoted by $\{T^{(k)}(U,D)\}_{k\in\mathbb{N}}$, any assignment of a trajectory to each state.
We call the map $T(U,D)$, associating to each state its trajectory, the trajectory-map.
\end{definition}
The physical meaning of the trajectory-map is the following: $T^{(k)}(U,D)([u])$ is a notation to indicate the state observed in a particular, physically admissible forward history at the end of the $k$-th time unit after one has prepared the state $[u]$ at time 0.
We may also indicate the trajectory of the state $[u]$ simply by $$T([u]):=\{T^{(k)}([u])\}_{k>0},$$ if no confusion can arise. 

Taken together, the collection of trajectories $\{T([u]),[u]\in \mathbb{P}(\mathcal{H})\}$ form a (physically admissible but otherwise totally arbitrary) sample realization of the collapse dynamics. 
As $D$ varies in the set $\mathfrak{S}$ of all admissible selectors, we get the mathematical description of every physically admissible discrete-time forward-history of collapse outcomes for every prepared state. As said in the Introduction, we want to establish a result that holds for any possible realization of the collapse dynamics.

We may denote $T^{(k)}(U,D)$ by $T^{(k)}_{U,D}$, or simply $T^{(k)}$ when the linear operator $U$ and the selector $D$ are clear from the context. 

\begin{remark}\label{no_dyn_} Let us observe explicitly that $T^{(k)}$ is not, in general, the iteration of order $k$ of some map $T$ defined on the state space $\mathbb{P}(\mathcal{H})$. Evaluating $T^{(1)}$, for instance, on the intermediate state $T^{(1)}([u])$ would also re-evaluate the sequence. This entails, in particular, that $T^{(k)}([u])=T^{(h)}([v])$ does not imply that $T^{(k+1)}([u])=T^{(h+1)}([v])$.

This mathematical fact reflects the physical fact that the outcome of collapse events are treated as intrinsically non-deterministic, so that even recording the full forward history of the collapse events seen by a certain prepared state $[u]$ does not tell us anything about what would happen to the state $T^{(k)}([u])$ if it were obtained along the evolution of another state $[v]$ in the same open system.
\end{remark}
\begin{remark}\label{dyn_}
If the trajectories are compatible, in the sense that, for every $h,k\in\mathbb{N}$ and every $[u],[v]\in \mathbb{P}(\mathcal{H})$, we have: $$T^{(k)}([u])=T^{(h)}([v])\implies T^{(k+1)}([u])=T^{(h+1)}([v]),$$
then, setting $g([u])=T^{(1)}([u])$, we have a particular case in which the sample realization of the open system can be identified with a (generally discontinuous) topological dynamical system, $(\mathbb{P}(\mathcal{H}),g)$, whose forward orbits coincide with the trajectories $T^{(k)}$.
\end{remark}

In a general open quantum system, none of the measurement channels (including the ``blank''
no-collapse case) can be regarded as strictly forbidden in any open
region of the state space.  
More precisely, every outcome label $j\in\Lambda$ must
remain topologically accessible in arbitrarily small neighborhoods of
every pure state.  Therefore, each outcome label $j\in\Lambda$ occurs arbitrarily close
to every pure state, which leads to the topological condition
\begin{equation}\label{densityproperty}
\overline{\{ [u] \in \mathbb{P}(\mathcal{H}) : D([u])_0=j\}} = \mathbb{P}(\mathcal{H})
\qquad \text{for all } j\in\Lambda.
\end{equation}
Note that we are not assuming \eqref{densityproperty}, but we stress that we are \emph{allowing} it.  Note that, under this topological condition, $T^{(k)}$ is discontinuous everywhere for every $k$. 

\begin{lemma}\label{lem:open-dense-Dj}
If $P_j$ is not the zero operator for $j\neq 0$, the sets
$$
D_j=\{[u]\in\mathbb{P}(\mathcal{H}):\|P_j u\|\neq 0\}
$$
are open for all $j\in\Lambda$, and $D_j$ is dense for every $j\neq 0$. In particular $D_0=\mathbb{P}(\mathcal{H})$.
\end{lemma}

\begin{proof}
Openness follows from continuity of $[u]\mapsto \|P_j u\|$. If $j\neq 0$, then $P_ju\neq 0$, so $\ker(P_ju)$ is a proper subspace of $\mathcal{H}$. 

Consequently, the projectivization $\mathbb{P}(\ker(P_j))$ is a proper (closed) projective subspace of $\mathbb{P}(\mathcal{H})$, hence it has empty interior. Therefore, $$
D_j=\mathbb{P}(\mathcal{H})\setminus \mathbb{P}(\ker(P_j))$$ is dense. Finally, $P_0=I$ gives $D_0=\mathbb{P}(\mathcal{H})$.
\end{proof}

\bigskip

\begin{remark}
Note that the existence of an admissible selector does not conflict with the Kochen-Specker (KS) no-go theorem \cite{kochen2011problem}. In fact, it addresses a completely different point. KS forbids noncontextual, dispersion-free value assignments to \emph{all} projections simultaneously (in dimension $\ge 3$), preserving functional relations across incompatible observables. 
By contrast, throughout we fix a \emph{single} measurement context $\{P_j\}_{1\leq j\leq N}$ (one PVM), and the selector $D:\mathbb{P}(\mathcal{H})\to\Omega$ assigns an outcome sequence only for this context. 
The assignment is explicitly contextual/history-dependent and makes no claims about outcomes for other, incompatible PVMs. Hence $D$ is not a global valuation on the projection lattice, and the hypotheses of KS are not met.
\end{remark}

\begin{remark}
No non-degeneracy condition is imposed.
Even if every nonblank projection $P_j$, $j\neq0$ has rank $1$, the blank outcome $0$ interleaves the linear evolution
$[u]\mapsto [Uu]$ with discontinuous point-collapses.
Under \eqref{densityproperty}, $D([U\cdot])_0$ is nowhere locally constant; hence (for each $k$) $T^{(k)}$ is nowhere continuous.
Thus the chain-recurrence conclusion cannot be obtained trivially even in the rank-one collapse case.
\end{remark}

We want now to define a suitable version of strong chain-recurrence for our setting, in which the selector is not necessarily compatible. For this reason, while in ordinary dynamical systems the usual definition of (strong) chain relation covers also the case of finite orbits, here we must treat separately the case of finite trajectories and infinite trajectories.

\begin{definition}
    We say that a state $[u]\in\mathbb{P}(\mathcal{H})$ is \emph{periodic} of period $k\in\mathbb{N}$, if there exist $[x]\in \mathbb{P}(\mathcal{H})$ and $h\ge0$ such that 
    $[u]=T_{U,D}^{(h+k\cdot n)}([x])$ for all $n\in\mathbb{N}$.\\
    We say that $[u]$ is \emph{weakly periodic} if there exist $[x]\in \mathbb{P}(\mathcal{H})$ and a strictly increasing sequence $\{k_j\}_{j\in\mathbb{N}}$ of natural numbers such that $[u]=T^{(k_j)}_{U,D}([x])$ for all $j\in\mathbb{N}$.
\end{definition}
\begin{definition}\label{strong_ch_}
Given two states $[x],[y]\in \mathbb{P}(\mathcal{H})$ and $\epsilon>0$, we say that a finite sequence of points $[x_0],[x_1],\ldots,[x_n] $ of $\mathbb{P}(\mathcal{H})$, with $n\in\mathbb{N}$, is a \emph{strong $(\epsilon,d_{FS})$-chain} from
$[x]$ to $[y]$ if
\begin{enumerate}
    \item[i)] $[x_0] = [x]$ and $[x_n] = [y]$,
    \item [ii)] there are $h_0,\ldots,h_{n-1} $ natural numbers such that $\sum_{i=0}^{n-1} d_{FS}(T_{U,D}^{(h_i)}([x_i]),[x_{i+1}]) <\epsilon$.
\end{enumerate}
The \emph{strong chain relation} $\mathcal{SC}_{d_{FS}}\subseteq \mathbb{P}(\mathcal{H})^2$ is the binary relation defined as follows: given $[x], [y] \in \mathbb{P}(\mathcal{H})$,
$$
[x]\, \mathcal{SC}_{d_{FS}}\, [y] \iff \forall \epsilon>0\text{ there exists a strong $(\epsilon,d_{FS})$-chain from $[x]$ to $[y]$}.
$$
We say that a state $[u]$ is \emph{strong chain-recurrent}, if either $[u]$ is weakly periodic or $[u]\,\mathcal{SC}_{d_{FS}}\,[u]$.
\end{definition}

\begin{lemma}
    The strong chain relation $\mathcal{SC}_{d_{FS}}$ is transitive. 
    \end{lemma}
    \begin{proof}
    Let $[u],[v],[w]\in \mathbb{P}(\mathcal{H})$ be such that $[u]\,\mathcal{SC}_{d_{FS}}\,[v]$ and $[v]\,\mathcal{SC}_{d_{FS}}\,[w]$. 
    Then, for every $\epsilon>0$, let $[u]=[y_0],[y_1],\ldots,[y_n]=[v]$ ($n\in\mathbb{N}$) with $h_0,\ldots,h_{n-1}$ in $\mathbb{N}$ be a strong $(\epsilon/2,d_{FS})$-chain from $[u]$ to $[v]$ and let $[v]=[z_0],[z_1],\ldots,[z_m]=[w]$ ($m\in\mathbb{N}$) with $k_0,\ldots,k_{m-1}$ in $\mathbb{N}$ be a strong $(\epsilon/2,d_{FS})$-chain from $[v]$ to $[w]$.
    Therefore, 
    $[u]=[y_0],[y_1],\ldots,[y_{n-1}],[v],[z_1],\ldots,[z_{m-1}],[w]$
    is a strong $(\epsilon,d_{FS})$-chain from $[u]$ to $[w]$, because, since $[y_n]=[v]=[z_0]$, we have
    \begin{equation*}
    \begin{aligned}
        &\sum_{i=0}^{n-2}d(T^{(h_i)}([y_i]),[y_{i+1}])+
        d(T^{(h_{n-1})}([y_{n-1}]),[v])+
        d(T^{(k_0)}([v]),[z_1])+
        \sum_{i=1}^{m-1}d(T^{(k_i)}([z_i]),[z_{i+1}])=\\ &= \sum_{i=0}^{n-1}d(T^{(h_i)}([y_i]),[y_{i+1}])+
        \sum_{i=0}^{m-1}d(T^{(k_i)}([z_i]),[z_{i+1}])<
        \frac \epsilon 2 + \frac \epsilon 2=\epsilon.
    \end{aligned}
    \end{equation*}
    \end{proof}

    We give now the definition of $\mathcal{SC}_{d_{FS}}$-transitivity and internal $\mathcal{SC}_{d_{FS}}$-transitivity for open systems, which is an adaptation of the same definition for ordinary dynamical systems (Defs. \ref{def_strong_trans} and \ref{def_int_s_chain}). Since in the case of weakly periodic points the chains are not defined, we consider this case as vacuously verifying the internal transitivity condition.  

    \begin{definition}\label{intern_chain_tr_open}
        We say that a subset $S\subseteq \mathbb{P}(\mathcal{H})$ of the state space is $\mathcal{SC}_{d_{FS}}$-transitive if, for every $[u],[v]\in S$, we have $[u]\,\mathcal{SC}_{d_{FS}}\, [v]$ and $[v]\,\mathcal{SC}_{d_{FS}}\,[u]$.

        We say that $S\subseteq \mathbb{P}(\mathcal{H})$ is \textit{internally} $\mathcal{SC}_{d_{FS}}$-transitive if either $S$ consists of one weakly periodic state or for every $[u],[v]\in S$ and for every $\epsilon>0$ there exists a strong $(\epsilon,d_{FS})$-chain from $[x]$ to $[y]$ whose points belong to $S$.
    \end{definition}

\section{Reversibility} \label{sec_4}

\begin{definition}[Integrated energetic cost]\label{def_cost_}
Let $H_{\text{eff}}$ be the effective Hamiltonian and let a bounded self-adjoint \emph{perturbation} $\Delta H$ act on the segments between
collapse events, so that the evolution obeys: 
$$
i\,\partial_t W(t)=(H_{\text{eff}}+\Delta H)W(t),
$$ 
for a suitable operator $W$.
For a time interval $I\subset\mathbb{R}$, and indicating by $[u_t]$ the state of the system at time $t$, the \emph{integrated energetic cost} of
$\Delta H$ on $I$ is
$$
\int_I \|\Delta H u_t\|\,dt\le \int_I \|\Delta H\|_*\,dt=:\mathcal{E}_I[\Delta H], 
$$
where $\|\cdot\|_*$ denotes any unitarily invariant matrix norm (by default, the
operator norm). 
\end{definition}

\begin{definition}\label{def_energ_cost}
Let $\gamma$ be a finite forward itinerary given by:
$$\gamma = \big([u_0],[u_1],\dots,[u_N]\big),
   \qquad [u_{k+1}] = T^{(h_k)}_{W_k,D}([u_k]),$$
where for every $k=0,\ldots,N-1$ let $\Delta H_k$ be a bounded self-adjoint \emph{perturbation}, such that the evolution obeys 
$$
i\,\partial_t W_k(t)=(H_{\text{eff}}+\Delta H_k)W_k(t)
$$ 
for a suitable operator $W_k$.
Then, the \emph{energetic cost} of the itinerary $\gamma$ is defined as:
$$\mathcal{E}_D(\gamma)=\sum_{k=0}^{N-1} \mathcal{E}_{I_k}[\Delta H_k], $$
where $\mathcal{E}_{I_k}[\Delta H_k]$ is the upper bound of the integrated energetic cost of $\Delta H_k$ on $I_k$ in Def.~\ref{def_cost_}.
\end{definition}

For $\epsilon>0$, we denote by $B_\epsilon([Uu])$ the open $d_{FS}$-ball of radius $\epsilon$ centered at $[Uu]$ in $\mathbb{P}(\mathcal{H})$.

Let us prove the following Lemma:
\begin{lemma}
\label{pert_2}
Let $u \in \mathcal{H}$ and $0<\tau_0<\tau_1$.  
For every (small enough) $\epsilon>0$ and $[v] \in B_\epsilon([U(\tau_1,\tau_0)u])$, there exist $\tau \in (\tau_0,\tau_1)$ and a self-adjoint perturbation $\Delta H_\epsilon$ such that the operator
$$
V(\tau_1,\tau)=e^{-i (H_{\text{eff}}+\Delta H_\epsilon)(\tau_1-\tau)}
$$
satisfies
$$
[V(\tau_1,\tau)\,U(\tau,\tau_0) u]= [v],
\qquad\text{and}\qquad
\lim_{\epsilon\to 0} \| V(\tau_1,\tau) - U(\tau_1,\tau)\|=0.
$$
\end{lemma}

\begin{proof}
Set $w :=U(\tau_1,\tau_0) u$ and $\widehat{w} := \widehat{U}(\tau_1,\tau_0) u=\frac{w}{\|w\|}$. 
Without loss of generality, let us assume $\langle \widehat{w}, v \rangle \geq 0$ (replacing $v \mapsto e^{-i\arg\langle \widehat{w},v\rangle} v$ if necessary) in such a way to preserve $\|v\|=1$ and $d_{\mathrm{FS}}([\widehat{w}],[v])$.
Set $\Delta\tau\coloneqq\tau_1-\tau>0$ and let $A_u: \mathcal{B}_{sa}\to \mathcal{H}$ be a map from the set of self-adjoint operators, denoted by $\mathcal{B}_{sa}$, to $\mathcal{H}$ defined as
$$
A_u(J)=e^{-i(H_\text{eff}+J)\Delta\tau} U(\tau,\tau_0)u.
$$
Notice that $A_u(0)=U(\tau_1,\tau)U(\tau,\tau_0)u=w$.
Let $J=t K$ with $K\in \mathcal{B}_{sa}$.
We want to evaluate 
$$
D A_u(0)[K]:=\frac{d}{dt} A_u(t K)\Big|_{t=0}=\frac{d}{dt} e^{-i(H_{\text{eff}}+t K)\Delta\tau}U(\tau,\tau_0)u \Big|_{t=0}.
$$
By \cite[Theorem 3(b)]{haber2018notes} we find
\begin{equation}\label{eq_der_A}
    D A_u(0)[K]=-i\Delta\tau\int_0^1
    e^{-i(1-s)H_{\text{eff}}\Delta\tau} K e^{-isH_{\text{eff}}\Delta\tau} ds\; U(\tau,\tau_0) u.
\end{equation}
Set $N(J):=\| A_u(J) \|=\sqrt{\langle A_u(J),A_u(J)\rangle}$ and so $N(t K)^2=\langle A_u(t K),A_u(t K)\rangle$.
Then 
\begin{equation}\label{eq_dN}
    \frac{d}{dt} N(t K)^2=\langle \frac{d}{dt}A_u(t K),A_u(t K)\rangle+ \langle A_u(t K), \frac{d}{dt}A_u(t K)\rangle = 2 Re \langle A_u(t K), \frac{d}{dt}A_u(t K)\rangle.
\end{equation}
Since $\frac{d}{dt} N(t K)^2=2N(t K) \frac{d}{dt}N(t K)$, by Eq.~\eqref{eq_dN} and evaluating the derivative at $t=0$, we find that $N(0)DN(0)[K]=Re\langle A_u(0), D A_u(0)[K]\rangle$. 
Therefore,
\begin{equation}\label{eq_der_N}
   DN(0)[K]= \frac{Re\langle A_u(0), D A_u(0)[K]\rangle}{N(0)}=\frac{Re\langle w, D A_u(0)[K]\rangle}{\|w\|}.
\end{equation}
Let $G_u$ be the map given by $G_u(J)=\frac{A_u(J)}{N(J)}$ for $J\in \mathcal{B}_{sa}$. 
Notice that $G_u(J)$ is a unit vector.
Now we want to evaluate the derivative of $G_u$ at 0. 
By \eqref{eq_der_A} and \eqref{eq_der_N}, we have that 
\begin{equation}\label{DG_}
    D G_u(0)[K]=\frac{D A_u(0)[K]\, N(0)-A_u(0)\, DN(0)[K]}{N(0)^2}=\frac{D A_u(0)[K]}{\|w\|}-\frac{w \,Re\langle w, D A_u(0)[K]\rangle}{\|w\|^3}
\end{equation}

We want to establish the surjectivity of $DG_u(0)[\cdot]$ (in so doing we can of course ignore the term parallel to $w$). We have:
\begin{equation}
   D A_u(0)[K]=-i\Delta\tau  \int_0^1
    e^{-i(1-s)H_{\text{eff}}\Delta\tau} K e^{-isH_{\text{eff}}\Delta\tau} ds\; U(\tau,\tau_0)u=-i \Delta\tau K U(\tau,\tau_0)u+\mathrm{O}(\Delta \tau^2)\|K\|,
\end{equation}

therefore $$P_{\widehat{w}}DG_u(0)[K]=-\frac{i\Delta\tau}{\|w\|}P_{\widehat w}K U(\tau,\tau_0)u+O(\Delta\tau^2)\|K\|. $$

The leading map $K\mapsto P_{\widehat w}K U(\tau,\tau_0)u$ is surjective onto $\widehat{w}^\perp$. Indeed, for $\Delta\tau$ small enough and for every $z\in \widehat{w}^\perp$ let us define the self-adjoint operator $K_z$ as
\begin{equation*}
    K_z\coloneqq \frac{\ket{z}\bra{\widehat{w}}}{\langle{\widehat{w},U(\tau,\tau_0)u\rangle}}+\frac{\ket{\widehat{w}}\bra{z}}{\overline{\langle{\widehat{w},U(\tau,\tau_0)u}\rangle}}.
\end{equation*}
Applying $K_z$ to $U(\tau,\tau_0)u$ yields
\begin{equation*}
    K_zU(\tau,\tau_0)u=\frac{\langle \widehat{w}, U(\tau,\tau_0)u\rangle}{\langle{\widehat{w},U(\tau,\tau_0)u\rangle}}z+ \frac{\langle z, U(\tau,\tau_0)u\rangle}{\overline{\langle{\widehat{w},U(\tau,\tau_0)u}\rangle}}\widehat{w}=z+ \frac{\langle z, U(\tau,\tau_0)u\rangle}{\overline{\langle{\widehat{w},U(\tau,\tau_0)u}\rangle}}\widehat{w}.
\end{equation*}
Finally, after applying $P_{\widehat w},$ we obtain $P_{\widehat w}K_z U(\tau,\tau_0)u=z$ and therefore, for $\Delta\tau$ small enough, $P_{\widehat w}DG_u(0)[K]$ is also surjective onto $\widehat{w}^\perp$.

Therefore, the derivative is surjective onto the tangent space \(T_{[\widehat w]}\mathbb P(\mathcal H)\). Therefore (recalling that $\pi$ indicates the canonical projection) the map $\pi\circ G_u$ is open near 0, and thus there exists a neighborhood $\mathcal{J}$ of $0$ in $\mathcal{B}_{sa}$ such that
$B_\epsilon([\widehat{w}])\subseteq \pi \circ G_u(\mathcal{J})$. 
Moreover, by the Local Section Theorem (see \cite[Theorem 4.26]{lee2013smooth}), the map $\pi\circ G_u$ admits a smooth local section. Hence, after possibly shrinking $\mathcal{J}$, this section is locally Lipschitz, and one can choose the preimage $J\in\mathcal{J}$ of $[v]$ with
 $\|J\|\le C\|\widehat{w}-v\|\le C\epsilon$.

The integrated energetic cost of
$\Delta H_\epsilon:=J$ on $I:=[\tau,\tau_1]$ is therefore bounded from above by
\begin{equation}\label{eq_energ_cost}
    \mathcal{E}_I[\Delta H_\epsilon]= \int_I \|\Delta H_\epsilon\|\,dt\le C\epsilon\,\Delta\tau, 
\end{equation}
which means that $\mathcal{E}_I[\Delta H_\epsilon]=O(\epsilon)$ since $\Delta\tau$ and $C$ are constants.

Let $V(t,\tau)$ solve the Schr\"{o}dinger evolution equation $\,\partial_t V(t,\tau)=-i (H_{\text{eff}}+\Delta H_\epsilon )V(t,\tau)$ with $V(\tau,\tau)=I$.
From $\partial_t U(t,\tau) = -iH_{\text{eff}}U(t,\tau)$ we have that  $V(t,\tau) - U(t,\tau)$ satisfies the following differential equation:
$$
\partial_t \big(V(t,\tau) - U(t,\tau)\big )=-iH_{\text{eff}}\big(V(t,\tau) - U(t,\tau)\big )-i\Delta H_{\epsilon}V(t,\tau),
$$
with solution 
$$
V(t,\tau) - U(t,\tau)=-i\int_{\tau}^t U(t,s)\Delta H_{\epsilon}V(s,\tau)\ ds.
$$
Finally, exploiting the result in norm and recalling that $\|U(t,s)\|\le1$ for $t\ge s$ (see Remark~\ref{rem_cond_numb}) we have the following inequality
$$
\|V(t,\tau) - U(t,\tau)\|\le \|\Delta H_{\epsilon}\|\int_{\tau}^t \|V(s,\tau)\|\ ds\le C\epsilon \int_{\tau}^t \|V(s,\tau)\|\ ds,
$$
where $C$ does not depend on $\epsilon$ nor $J$. Furthermore,
\begin{equation}\label{eq_V}
\begin{aligned}
    \int_{\tau}^{\tau_1} \|V(s,\tau)\|\ ds&\le \int_\tau^{\tau_1} e^{\|H_{\text{eff}}+\Delta H_\epsilon\|(s-\tau)}\ ds\le 
    \frac{1}{\|H_{\text{eff}}\|+\|\Delta H_\epsilon\|} \Big( e^{(\|H_{\text{eff}}\|+\|\Delta H_\epsilon\|)\Delta\tau}-1 \Big)\le\\
    &\le \frac{1}{\|H_{\text{eff}}\|}\Big( e^{(\|H_{\text{eff}}\|+C\epsilon)\Delta\tau}-1\Big)=:M(\Delta\tau,\epsilon)
\end{aligned}
\end{equation}
which is bounded as $\epsilon\to 0$.
By \eqref{eq_energ_cost} and \eqref{eq_V}, it follows that
$$
\|V(\tau_1,\tau) - U(\tau_1,\tau)\|\le C \epsilon\,M(\Delta\tau,\epsilon),
$$
with $M(\Delta\tau,\epsilon)$ bounded as $\epsilon\to 0$,
from which the claim follows.
\end{proof}

The following result is key to our reasoning. For the proof, see Section~\ref{sec_proof}. 
\begin{theorem}\label{th_rec_}
Let $U$ be a linear operator and $D$ an admissible selector. 
    Then, the open system $(\mathbb{P}(\mathcal{H}),\{T_{U,D}^{(k)}\}_{k\in\mathbb{N}})$ has a topologically closed, internally strong chain-transitive subset.
\end{theorem}

The following result establishes that we can construct chains in $(\mathbb{P}(\mathcal{H}),\{T^{(k)}_{U,D}\}_{k\in\mathbb{N}})$ with self-adjoint perturbations having arbitrarily small energetic cost. 

\begin{theorem}\label{th_rec_2}
For every pair $[u],[v]\in \mathbb{P}(\mathcal{H})$ with $[u]\,\mathcal{SC}_d\,[v]$ and every $\epsilon>0$,
there exist a finite family of self-adjoint perturbations
$\{\Delta H_{\epsilon_k}\}_{k=1}^{N-1}$ and a family of time intervals $\{\Delta \tau_k\}_{k=1}^{N-1}$ satisfying
$$
\sum_{k=1}^{N-1} \mathcal{E}_{\Delta \tau_k} [\Delta H_{\epsilon_k}] = O\Big(\sum_{k=1}^{N-1}\epsilon_k\Big)=O(\epsilon)
$$
and such that $d_{FS}(T^{(h)}_{W,D}\,[u],[v])<\epsilon$ for a certain $h\in\mathbb{N}$ depending on $\epsilon$, where $W$ is defined by the time-ordered exponential:
$$
W(t):=\mathcal{T}\exp\Big(-i\int_{0}^{t} H_{\mathrm{tot}}(s)\,ds\Big)
\qquad\text{with}\qquad
H_{\mathrm{tot}}(t):=H_{\text{eff}}+\sum_{k=1}^{N-1}\mathbf 1_{\Delta \tau_k}\,\Delta H_{\epsilon_k}.
$$
\end{theorem}

\begin{proof}
Fix $\epsilon>0$. 
Since $[u]\,\mathcal{SC}_{d_{FS}}\,[v]$, there is a strong $\epsilon$-chain
$$
[u_0]=[u],[u_1],\dots,[u_{N-1}], [u_N]=[v]
$$ 
with
$d_{FS}(T^{(h_k)}_{U,D}([u_k]),[u_{k+1}])<\epsilon_k$ for $k=0,\dots,N-1$ where $\{h_0,\ldots,h_{N-1}\}\subseteq\mathbb{N}$ and
$\sum_{k=0}^{N-1}\epsilon_k<\epsilon$.

For every $k=1,\dots,N$ we set
$$
s_k:=\sum_{j=0}^{k-1}h_j,\quad [w_k]:=T^{(h_{k-1})}_{U,D}([u_{k-1}]),
\quad \tau_k^+:=s_k+\tfrac{2}{3},
$$
For each $k=1,\dots,N-1$, we can apply Lemma~\ref{pert_2} with
$$
\tau_0=s_k,\quad \tau_1=\tau_k^+,\quad u=w_k,\quad v=U(\tau_k^+,s_k)\,u_k,
$$
and $\epsilon=L_k\,\epsilon_k$ where $L_k$ is the Lipschitz constant of the map $[u]\mapsto [U(\tau_1,\tau_0)u]$ (see Remark~\ref{rem_lips}), because
$$
d_{FS}([U(\tau_1,\tau_0)\,u],[v])=d_{FS}([U(\tau_k^+,s_k)\,w_k], [U(\tau_k^+,s_k)\,u_k])\le L_k\,d_{FS}([w_k],[u_k])<L_k\,\epsilon_k.
$$
Therefore, there exist $\tau_k^-\in (s_k,\tau_k^+)$ and a perturbation $\Delta H_{\epsilon_k}$ supported on the interval $\Delta\tau_k:=[\tau_k^-,\tau_k^+]$.
The associated propagator
$$
V_k(\tau_k^+,\tau_k^-):= e^{-i(H_{\text{eff}}+\Delta H_{\epsilon_k})\Delta\tau_k}
$$
then satisfies
$$
[V_k(\tau_k^+,\tau_k^-)\,U(\tau_k^-,s_k)\,w_k] = [U(\tau_k^+,s_k)\,u_k],
\qquad
\|V_k(\tau_k^+,\tau_k^-)-U(\tau_k^+,\tau_k^-)\|\xrightarrow[\epsilon_k\to 0]{}0.
$$
In particular, if we set 
$$
W_k(t,s_k):=U(t,\tau_k^+)\,V_k(\tau_k^+,\tau_k^-)\,U(\tau_k^-,s_k),$$ for $t\in[\tau_k^+, s_{k+1}]$, we find
$$
[W_k(t,s_k) \,w_k]= [U(t,\tau_k^+)U(\tau_k^+,s_k)\,u_k ]= [U(t,s_k)\,u_k]\qquad \text{for }t\in[\tau_k^+, s_{k+1}].
$$
Therefore, 
\begin{equation}\label{eq_wk_1}
\begin{aligned}
    T^{(1)}_{W_k,D}([w_k])&=f_{D([W_k(s_k+1,s_k)w_k])_0}([W_k(s_k+1,s_k)w_k])=\\
    &=f_{D([Uu_k])_0}([Uu_k])=T^{(1)}_{U,D}([u_k]).
\end{aligned}
\end{equation}
By Eq.~\eqref{eq_wk_1} and since $W_k$ coincides with $U$ from time $\tau_k^+$ to  $s_{k+1}$, we have that:
\begin{equation}\label{eq_wk}
     T^{(h_k)}_{W_k,D}([w_k])=T_{U,D}^{(h_k)}([u_k])=[w_{k+1}].
\end{equation}
Define the total evolution over the time-interval $[0,s_N]$ by:
$$
W(t)=\mathcal{T}\exp\big(-i\int_0^t H_{\mathrm{tot}}(s)\,ds\big)
$$ 
with
$H_{\mathrm{tot}}(t)$ as defined in the statement.
If we take $t=s_N$, we can equivalently write the total evolution 
with perturbations applied at intervals $[\tau^-_k,\tau^+_k]$ \( (k=1, \ldots, N-1) \) as the left-ordered product:
$$
W(s_N,0)=\overleftarrow{\prod_{k=1}^{N-1}}
\Big[W_k(s_{k+1},s_k)\Big]U(h_0,0)=\overleftarrow{\prod_{k=1}^{N-1}}
\Big[\,U(s_{k+1},\tau_k^+)\,V_k(\tau_k^+,\tau_k^-)\,U(\tau_k^-,s_k)\,\Big]U(h_0,0).
$$
Since $W(h_0,0)=U(h_0,0)$ and $W(s_{k+1},s_k)=W_k(s_{k+1},s_k)$ for $k=1,\ldots,N-1$, by~\eqref{eq_wk} it follows that 
$$
T^{(s_k)}_{W,D}([u]) = T^{(h_{k-1})}_{U,D}([u_{k-1}])
$$
for every $k=1,\ldots,N$. 
Therefore, 
\begin{equation}
    d_{FS}\big(T^{(s_N)}_{W,D}([u]),[v]\big)=d_{FS}\big(T_{U,D}^{(h_{N-1})}([u_{N-1}]),[v]\big)<\epsilon_{N-1}<\epsilon.
\end{equation}
Notice that $[W(s_N,0)u]=[U(h_{N-1},0)u_{N-1}]$ if, for every $k=0,\ldots,N-1$ we have that $D([U u_k])_j=0$ for all $j\in\{0,\ldots,h_k-1\}$.

Finally, the integrated energetic cost of $\Delta H_{\epsilon_k}$ on $\Delta\tau_k$ is bounded from above by
$$
\mathcal{E}_{\Delta\tau_k} [\Delta H_{\epsilon_k}]=\int_{\tau_k^-}^{\tau_k^+}\|\Delta H_{\epsilon_k}\|\,dt\le C_k L_k \,\epsilon_k\,\Delta\tau_k.
$$
Hence $\mathcal{E}_{\Delta\tau_k} [\Delta H_{\epsilon_k}]=O(\epsilon_k)$ and $\sum_{k=1}^{N-1} \mathcal{E}_{\Delta\tau_k}[\Delta H_{\epsilon_k}]=\sum_{k=1}^{N-1}\int\|\Delta H_{\epsilon_k}\|=O(\sum_k\epsilon_k) =O(\epsilon)$.
\end{proof}

\vspace{0.3cm}

Let us finally address briefly a natural question. The subset $S$ showing internal strong chain-transitivity can be very small, up to just one state. But how dynamically relevant are these subsets? The answer is ``very relevant, no matter how small they are". Indeed, \textit{every} state is strong chain-related to at least one of them. For instance, in the extreme case in which there is just one chain-transitive subset and it contains just one state, that state is then very special, being reachable with vanishing energetic cost from \textit{any} other state. 

This is made precise below.

\begin{definition}
  Let $S\subseteq \mathbb{P}(\mathcal{H})$ be a closed and internally $\mathcal{SC}_{d_{FS}}$-transitive subset of the state space.  
  The \emph{chain basin of attraction} of $S$ is the set of states $[u]$ such that $[u]$ is $\mathcal{SC}_{d_{FS}}$-related to a state in $S$.
\end{definition}

\begin{theorem}
    The set of all chain basins of attraction covers $ \mathbb{P}(\mathcal{H})$. 
\end{theorem}
\begin{proof}
By Corollary \ref{cor_basin} in Section \ref{sec_proof}, any state $[u]\in \mathbb{P}(\mathcal{H})$ is $\mathcal{SC}_{d_{FS}}$-related to a closed, internally strong chain-transitive subset $S\subseteq \mathbb{P}(\mathcal{H})$, from which the thesis follows immediately.
\end{proof}

\section{Proof of Theorem \ref{th_rec_}}\label{sec_proof} 

In this section, we address the proof of Theorem \ref{th_rec_}. 
This result is proved using similar ideas to those used in the proof of Theorem~3.9 in~\cite{della2025chains}. However, there are differences since in this case we do not have an ordinary dynamical system (an $\mathbb{N}$-action on $\mathbb{P}(\mathcal{H})$), because the trajectories of different states are in general not compatible, as already observed in Remark \ref{no_dyn_}. Moreover, we need to show strong chain-recurrence and not only chain-recurrence, although this is a lesser challenge. As stated in the Introduction, we aim to provide an argument that does not rely on the Axiom of Choice, so that the construction makes the link between the chain-transitive subset and the dynamics transparent.

We start by proving the following result, which shows that, for a linear operator $U$ and an admissible selector $D$, in the open system $(\mathbb{P}(\mathcal{H}),\{T_{U,D}^{(k)}\}_{k\in\mathbb{N}})$ there always exists a strong chain-recurrent state.
For a standard introduction to ordinals and transfinite induction, we refer to \cite{jech2003set}, Chapter 2.

\begin{lemma}\label{thm_existence}
    Let $U$ be a linear operator and $D$ an admissible selector. 
    Then, the open system $(\mathbb{P}(\mathcal{H}),\{T_{U,D}^{(k)}\}_{k\in\mathbb{N}})$ has a strong chain-recurrent point.
\end{lemma}

\begin{proof}
Suppose that there exists $[u]\in \mathbb{P}(\mathcal{H})$ such that its trajectory $T([u])=\{T^{(k)}([u])\}_{k>0}$ is a finite set of states.
Notice that the sequence $D([u])\in\Omega$ is not necessarily periodic.
Since the trajectory of $[u]$ is finite, there exist a state $[v]\in T([u])$ and a sequence $\{k_j\}_{j\in\mathbb{N}}$ such that $[v]=T^{(k_j)}([u])$ for all $j\in\mathbb{N}$. 
Therefore, $[v]$ is a weakly periodic point, and, in particular, $[v]$ is a strong chain-recurrent point.
So we assume that the trajectory of every state in $\mathbb{P}(\mathcal{H})$ is an infinite set.

Pick $[u_0]\in \mathbb{P}(\mathcal{H})$ 
and take $[u_1]\in (T([u_0]))'$, i.e. a point of the derived set of the trajectory $T([u_0])$, which is nonempty due to the compactness of $\mathbb{P}(\mathcal{H})$. 
Since $[u_1]\in (T([u_0]))'$, we have that for every $\epsilon>0$ there is $h>0$ such that $d_{FS}(T^{(h)}([u_0]),[u_1])<\epsilon$. 
Therefore, $[u_0]\,\mathcal{SC}_{d_{FS}}\,[u_1]$. 

Let us now define a set of points by transfinite induction as follows.
Let $\alpha>1$ be an ordinal number and assume that $[u_\beta]$ has been defined for every $0\le\beta<\alpha$ and set: 
\begin{equation}
    S_{\beta,\alpha}([u_0]):=\{[u_\eta]\}_{\beta\le \eta<\alpha} \quad , \quad
    S_\alpha([u_0]):=S_{0,\alpha}.
\end{equation}
\begin{itemize}
\item If $\alpha$ is a successor ordinal, we pick $[v]\in (T([u_{\alpha-1}]))'$ and we set $[u_\alpha]:=[v]$. 

\item If $\alpha$ is a limit ordinal, we set 
\begin{equation}
    S^*_\alpha([u]):=\bigcap_{\beta<\alpha} (S_{\beta,\alpha}([u]))',
\end{equation}
and take $[v]\in S^*_\alpha([u])$. 
Notice that $S^*_\alpha([u])\neq\emptyset$ (see, e.g. \cite[Th. 26.9, p. 169]{munkres2013topology}), since $|S_{\beta,\alpha}([u])|\ge \aleph_0$ for every $\beta<\alpha$ so that $(S_{\beta,\alpha}([u]))'$ is a closed and nonempty set. Then, we set $[u_\alpha]:=[v]$.
\end{itemize}

Let us prove that, for every ordinal $\alpha$,
\begin{equation}\label{eq_strong}
    [u_\beta]\,\mathcal{SC}_{d_{FS}}\,[u_\eta] \qquad\text{whenever}\quad 0\le \beta<\eta\le\alpha.
\end{equation}
Let us prove this property by transfinite induction on $\alpha$.
If $\alpha=1$, since $[u_0]\,\mathcal{SC}_{d_{FS}}\,[u_1]$, the claim is true.
Assume then that $\alpha>1$ and that \eqref{eq_strong} is true for every $\eta<\alpha$. 
If $\alpha$ is a successor ordinal, then $[u_\alpha]\in (T([u_{\alpha-1}]))'$ and with the same argument used for $[u_0]$ and $[u_1]$ we can see that $[u_{\alpha-1}]\,\mathcal{SC}_{d_{FS}}\,[u_\alpha]$.
Then, by the induction hypothesis and using the transitivity of $\mathcal{SC}_{d_{FS}}$, for $0\le\beta<\eta<\alpha$, we have 
$$
[u_\beta]\,\mathcal{SC}_{d_{FS}}\,[u_{\alpha-1}]\implies [u_\beta]\,\mathcal{SC}_{d_{FS}}\,[u_{\alpha}].
$$ 
If $\alpha$ is a limit ordinal, take $\beta<\alpha$.  
Then, for every $\epsilon>0$, there is $\beta<\eta<\alpha$ sufficiently large that 
$$
d_{FS}([u_\eta],[u_\alpha])<\frac \epsilon 2.
$$ 
By the induction hypothesis, there is a strong $(\frac \epsilon 2,d_{FS})$-chain $[x_0], \ldots, [x_{n}]$ from $[u_\beta]$ to $[u_\eta]$, so in particular we have that $\sum_{i=0}^{n-1}d_{FS}(T^{(h_{i})}([x_i]),[x_{i+1}])<\frac \epsilon 2$ for certain natural numbers $h_0,\ldots,h_{n-1}$.
Then, by the triangle inequality and recalling that $[x_n]=[u_\eta]$, we have that 
\begin{align*}
    \sum_{i=0}^{n-2}&d_{FS}(T^{(h_{i})}([x_i]),[x_{i+1}])+d_{FS}(T^{(h_{n-1})}([x_{n-1}]),[u_\alpha])\le\\
    &\le \sum_{i=0}^{n-1}d_{FS}(T^{(h_{i})}([x_i]),[x_{i+1}])+d_{FS}([u_\eta],[u_\alpha])< \frac \epsilon 2 + \frac \epsilon 2=\epsilon,
\end{align*}
which means that $[x_0],\ldots,[x_{n-1}],[u_\alpha]$ is a strong $(\epsilon,d_{FS})$-chain from $[u_\beta]$ to $[u_\alpha]$.
Therefore, $[u_\beta]\,\mathcal{SC}_{d_{FS}}\, [u_\alpha]$ for every $\beta<\alpha$ and we conclude.

Now, if for every $\eta<\alpha$ we have $[u_\eta]\notin S_\eta([u_0])$, then $[u_\beta]\ne [u_\gamma]$ for every distinct $\beta,\gamma<\alpha$, so that the application 
\begin{equation}\label{eq_y}
\eta\mapsto [u_\eta]\quad\quad (0\le\eta<\alpha)
\end{equation}
is a bijection between $\alpha=\{\eta:0\le\eta<\alpha\}$ and $S_\alpha([u_0])$, so
we have $|S_\alpha([u_0])|=|\alpha|$. 
Assuming $[u_\alpha]\notin S_\alpha([u_0])$ for every ordinal $\alpha$, by Hartogs' Lemma (see \cite{hartogs1915problem}; for a modern reference see Lemma 3.4 in \cite{jech2003set}) there would be an ordinal $\lambda$ so large that $|S_\lambda([u_0])|>|\mathbb{P}(\mathcal{H})|$, in the sense that there is no injection from $S_\lambda([u_0])$ into $\mathbb{P}(\mathcal{H})$. But this is impossible since $S_\alpha([u_0])\subseteq \mathbb{P}(\mathcal{H})$. 
Therefore, $[u_\alpha]\in S_\alpha([u_0])$ for some ordinal $\alpha$, so that: 
\begin{equation}\label{eq_x}
[u_\alpha]=[u_\beta]\text{ for some } \beta<\alpha.
\end{equation}
Then, using \eqref{eq_strong} with $\eta=\alpha$, we have that $[u_\alpha]$ is a strong chain-recurrent point for the system.
\end{proof}

\begin{remark}\label{rem_strong}
    For every $0\le \beta<\eta\le\alpha$ and every $\epsilon>0$, there is a strong $(\epsilon,d_{FS})$-chain from $[u_\beta]$ to $[u_\eta]$.
    Hence, if $[u_\alpha]=[u_\beta]$ for some $\beta<\alpha$, then by the transitivity of $\mathcal{SC}_{d_{FS}}$ we have that the set $S_{\beta,\alpha}([u])$ is $\mathcal{SC}_{d_{FS}}$-transitive, that is, for every $[v],[w]\in S_{\beta,\alpha}([u])$ and for every $\epsilon>0$ there exists a strong $(\epsilon,{d_{FS}})$-chain from $[v]$ to $[w]$.
    
    Therefore, if $[v]\in S_\alpha([u])$ and $w\in S_{\beta,\alpha}([v])$, then $[v]\,\mathcal{SC}_{d_{FS}} [w]$. 
    Notice that the other side is not true in general, because, for $\beta>0$, if we take $[v]=[u_\eta]$ with $0\le \eta<\beta$ and $[w]=[u_\gamma]$ with $\beta\le \gamma<\alpha$, then, in general, $[w] \,\cancel{\mathcal{SC}}_{d_{FS}} [v]$.
\end{remark}

In \cite{della2025chains}, it is shown that a compact dynamical system has a chain-transitive subsystem. 
Using Lemma~\ref{thm_existence} and Remark~\ref{rem_strong}, we want to prove that the system $(\mathbb{P}(\mathcal{H}),\{T_{U,D}^{(k)}\}_{k\in\mathbb{N}})$ has a closed and internally $\mathcal{SC}_{d_{FS}}$-transitive subset. The proof of this result requires a different treatment with respect to the result in \cite{della2025chains} since, as already observed (see Remark \ref{no_dyn_}), the open system $(\mathbb{P}(\mathcal{H}),\{T_{U,D}^{(k)}\}_{k\in\mathbb{N}})$ is not a dynamical system in the standard sense.

Note that meeting the requirement of being both \textit{topologically closed} (so stable with respect to passing to limit points) and \textit{internally} $\mathcal{SC}_{d_{FS}}$-transitive (so that the chains are all available ``from within" the subset) is the difficult part. For instance, in the proof of Lemma \ref{thm_existence} we already obtained an internally $\mathcal{SC}_{d_{FS}}$-transitive set ($S_{\beta,\alpha}([u])$), but it is not guaranteed to be closed, while on the other hand the set consisting only of the $\mathcal{SC}_{d_{FS}}$-recurrent point whose existence is established in Lemma \ref{thm_existence} is trivially $\mathcal{SC}_{d_{FS}}$-transitive but in general not internally $\mathcal{SC}_{d_{FS}}$-transitive. 

\vspace{0.5cm}

We are now ready to prove Theorem \ref{th_rec_}.

\begin{theorem*}\label{thmstrongsubsyst}
   \emph{\textbf{4.4}} Let $U$ be a linear operator and $D$ an admissible selector. 
    Then, the open system $(\mathbb{P}(\mathcal{H}),\{T_{U,D}^{(k)}\}_{k\in\mathbb{N}})$ has a topologically closed, internally strong chain-transitive subset.
\end{theorem*}
\begin{proof}
We need to show that there is a closed and internally $\mathcal{SC}_{d_{FS}}$-transitive subset of $\mathbb{P}(\mathcal{H})$. 
Suppose, towards a contradiction, that $(\mathbb{P}(\mathcal{H}),\{T_{U,D}^{(k)}\}_{k\in\mathbb{N}})$ does not have a closed and internally $\mathcal{SC}_{d_{FS}}$-transitive subset. 
In particular, we assume that every trajectory is infinite; otherwise, there exists a weakly periodic point, and the claim holds.

Let $[u^0]\in \mathbb{P}(\mathcal{H})$.
By the last step of the proof of Lemma \ref{thm_existence}, there exists a least ordinal $\alpha_0>0$ such that $[u^0_{\alpha_0}]\in S_{\alpha_0}([u^0])$, and so $[u^0_{\alpha_0}]=[u^0_{\beta_0}]$ for some $\beta_0<\alpha_0$. 
Moreover, the set $S_{\beta_0,\alpha_0}([u^0])$ is an internally $\mathcal{SC}_{d_{FS}}$-transitive subset. 
Under our assumption $S_{\beta_0,\alpha_0}([u^0])$ cannot be closed, then there exists some $[v]\in (S_{\beta_0,\alpha_0}([u^0]))'\setminus S_{\beta_0,\alpha_0}([u^0])$.
Take $[u]\in S_{\alpha_0}([u^0])$ and $\epsilon>0$. 
We want to show that there exists a strong $(\epsilon,d_{FS})$-chain from $[u]$ to $[v]$ whose points belong to $S_{\alpha_0}([u^0])\cup\{[v]\}$.
Since $[v]\in (S_{\beta_0,\alpha_0}([u^0]))'$, there is $\beta_0\le \eta<\alpha_0$ such that $d_{FS}([u^0_\eta],[v])<\epsilon/2$. 
By Remark \ref{rem_strong}, $[u]\,\mathcal{SC}_{d_{FS}}\,[u^0_\eta]$ and so there exists a strong $(\epsilon/2,d_{FS})$-chain $[u]=[y_0],[y_1],\ldots,[y_n]=[u^0_\eta]$ from $[u]$ to $[u^0_\eta]$ whose points belong to $S_{\alpha_0}([u^0])$.
By the triangle inequality,
\begin{equation}\label{eq_s1}
\begin{aligned}
    \sum_{i=0}^{n-2}d_{FS}(T^{(h_i)}([y_i]),[y_{i+1}])&+d_{FS}(T^{(h_{n-1})}([y_{n-1}]),[v])\le \\
    &\le\sum_{i=0}^{n-1}d_{FS}(T^{(h_i)}([y_i]),[y_{i+1}])+d_{FS}([u^0_\eta],[v])<\frac \epsilon 2 +\frac \epsilon 2=\epsilon,
\end{aligned}
\end{equation}
which means that $[u]=[y_0],[y_1],\ldots,[y_{n-1}],[v]$ is a strong $(\epsilon,d_{FS})$-chain from $[u]$ to $[v]$ whose points belong to $S_{\alpha_0}([u^0])\cup\{[v]\}$.

If $[v]\neq [u^0]$, we set $[u^1]:=[v]$.
Otherwise, if $[v]=[u^0]$, then $S_{\alpha_0}([u^0])$ is an internally $\mathcal{SC}_{d_{FS}}$-transitive subset. 
Since it cannot be a closed set, assume that there is some $[w]\in (S_{\alpha_0}([u^0]))'\setminus S_{\alpha_0}([u^0])$, and set $[u^1]:=[w] \,(\neq [u^0])$. 
With an analogous argument as that used in \eqref{eq_s1}, we find that for every $[u]\in S_{\alpha_0}([u^0])$ and for every $\epsilon>0$ there exists a strong $(\epsilon,d_{FS})$-chain from $[u]$ to $[u^1]$ whose points belong to $S_{\alpha_0}([u^0])\cup \{[u^1]\}$.

In summary, we have shown that the set $\{[u^0],[u^1]\}$ has the following properties:
\begin{itemize}
    \item[(a$_1$)] there exists a least ordinal $\alpha_0$ such that $[u^0_{\alpha_0}]=[u^0_{\beta_0}]$ for some $\beta_0<\alpha_0$ and  $S_{\beta_0,\alpha_0}([u^0])$ is an internally $\mathcal{SC}_{d_{FS}}$-transitive subset;
    
    \item[(b$_1$)] 
    for every $[u]\in S_{\alpha_0}([u^0])$ and for every $\epsilon>0$ there exists a strong $(\epsilon,d_{FS})$-chain from $[u]$ to $[u^1]$ whose points belong to $S_{\alpha_0}([u^0])\cup \{[u^1]\}$;
    
    \item[(c$_1$)] $[u^0] \neq [u^1]$.
\end{itemize}

We now proceed by transfinite induction. 
Let $\lambda>1$ be an ordinal number. 
Assume that $[u^\gamma]$ has been defined for every $0\le \gamma<\lambda$ and that, setting 
\begin{equation*}
    M_{\beta,\alpha}:=\bigcup_{\beta\le \gamma<\alpha} S_{\alpha_\gamma}([u^\gamma]) \quad \text{for $\beta<\alpha\le \lambda$},
\end{equation*}
the set $\{[u^\gamma]\}_{\gamma<\lambda}$ has the following properties:

\begin{enumerate}
\item[(a$_\lambda$)] for every $0\le \gamma<\lambda$, there exists a least ordinal $\alpha_\gamma$ such that $[u^\gamma_{\alpha_\gamma}]=[u^\gamma_{\beta_\gamma}]$ for some $\beta_\gamma<\alpha_\gamma$ and $S_{\beta_\gamma,\alpha_\gamma}([u^\gamma])$ is an internally $\mathcal{SC}_{d_{FS}}$-transitive subset;

\item[(b$_\lambda$)] 
for every $[u]\in S_{\alpha_\eta}([u^\eta])$ and for every $\epsilon>0$ there exists a strong $(\epsilon,d_{FS})$-chain from $[u]$ to $[u^{\xi}]$ whose points belong to $M_{\eta,\xi}\cup \{[u^{\xi}]\}$ whenever $\eta<\xi<\lambda$ ;

\item[(c$_\lambda$)] $[u^\eta] \neq [u^\xi]$ whenever $\eta\neq \xi$ with $\eta,\xi<\lambda$. 
\end{enumerate}

We now proceed with the definition of the state $[u^\lambda]$.
Notice that, assuming (b$_\lambda$), it is sufficient to prove property (b$_{\lambda+1}$) for $\xi=\lambda$.

\begin{itemize}
    \item[1.] If $\lambda$ is a successor ordinal, we consider the set $S_{\beta_{\lambda-1},\alpha_{\lambda-1}}([u^{\lambda-1}])$, which is internally $\mathcal{SC}_{d_{FS}}$-transitive by property (a$_\lambda$). 
    Under our assumption, it cannot be a closed set, so there exists 
    $$
    [v]\in (S_{\beta_{\lambda-1},\alpha_{\lambda-1}}([u^{\lambda-1}]))'\setminus S_{\beta_{\lambda-1},\alpha_{\lambda-1}}([u^{\lambda-1}]).
    $$
    First, we prove that for every $[u]\in S_{\alpha_\eta}([u^\eta])$ with $\eta<\lambda$ and for every $\epsilon>0$ there exists a strong $(\epsilon,d_{FS})$-chain from $[u]$ to $[v]$ whose points belong to $M_{\eta,\lambda}\cup \{[v]\}$.
    
    Fix $\eta<\lambda$ and $\epsilon>0$ and take $[u]\in S_{\alpha_\eta}([u^\eta])$. 
    With an analogous argument used in \eqref{eq_s1}, there is a strong $(\epsilon/2,d_{FS})$-chain $[y_0],[y_1],\ldots,[y_n]$ with $\{h_0,\ldots,h_{n-1}\}\subseteq \mathbb{N}$ from $[u^{\lambda-1}]\in S_{\alpha_{\lambda-1}}([u^{\lambda-1}])$ to $[v]$ whose points belong to $S_{\alpha_{\lambda-1}}([u^{\lambda-1}])\cup\{[v]\}$. 
    Moreover, by property (b$_\lambda$), there is a strong $(\epsilon/2,d_{FS})$-chain  $[z_0],[z_1],\ldots,[z_m]$ with $\{k_0,\ldots,k_{m-1}\}\subseteq\mathbb{N}$ from $[u]$ to $[u^{\lambda-1}]$ whose points belong to $M_{\eta,\lambda-1}\cup \{[u^{\lambda-1}]\}$.
    Therefore, 
    $[z_0],[z_1],\ldots,[z_{m-1}],[u^{\lambda-1}],[y_1],\ldots,[y_n]$
    is a strong $(\epsilon,d_{FS})$-chain from $[u]$ to $[v]$ whose points belong to $M_{\eta,\lambda-1}\cup S_{\alpha_{\lambda-1}}([u^{\lambda-1}])\cup\{[v]\}=M_{\eta,\lambda}\cup\{[v]\}$. 
    Indeed, $[z_0]=[u]$, $[y_n]=[v]$ and since $[z_m]=[u^{\lambda-1}]=[y_0]$, we have
    \begin{equation*}\label{eq_s2}
    \begin{aligned}
        &\sum_{i=0}^{m-2}d_{\text{FS}}(T^{(k_i)}([z_i]),[z_{i+1}])+ d_{\text{FS}}(T^{(k_{m-1})}([z_{m-1}]),[u^{\lambda-1}])+d_{\text{FS}}(T^{(h_0)}([u^{\lambda-1}]),[y_1])+\\
        &+\sum_{i=1}^{n-1}d_{\text{FS}}(T^{(h_i)}([y_i]),[y_{i+1}]) =\sum_{i=0}^{m-1}d_{\text{FS}}(T^{(k_i)}([z_i]),[z_{i+1}])+ \sum_{i=0}^{n-1}d_{\text{FS}}(T^{(h_i)}([y_i]),[y_{i+1}])< \\
        &<\frac \epsilon 2 + \frac \epsilon 2=\epsilon
    \end{aligned}
    \end{equation*}
    and we conclude. Consider now the following two cases.
    \begin{itemize}
        \item[•] If $[v]\notin \{[u^\gamma]\}_{\gamma<\lambda}$, then we set $[u^\lambda]:=[v]$ and property (b$_{\lambda+1}$) is satisfied.
        
        \item[•] Assume that $[v]=[u^{\xi_0}]$ for some $\xi_0<\lambda$.
        By (b$_\lambda$) and the property of $[v]$ shown above, it follows that $M_{\xi_0,\lambda}$ is internally $\mathcal{SC}_{d_{FS}}$-transitive. 
        Since this set cannot be closed, there is some $[w_1]\in (M_{\xi_0,\lambda})'\setminus M_{\xi_0,\lambda}$, so that $[w_1]\notin \{[u^\gamma]\}_{\xi_0\le \gamma<\lambda}$. 
        Since $M_{\xi_0,\lambda}$ is internally $\mathcal{SC}_{d_{FS}}$-transitive and $[w_1]\in (M_{\xi_0,\lambda})'$, by the triangle inequality we have that for every $[u]\in M_{\xi_0,\lambda}$ and for every $\epsilon>0$ there exists a strong $(\epsilon,d_{FS})$-chain from $[u]$ to $[w_1]$ whose points belong to $M_{\xi_0,\lambda}\cup \{[w_1]\}$.
        By property (b$_\lambda$) and using the transitivity of $\mathcal{SC}_{d_{FS}}$, we have that for every $[u]\in S_{\alpha_\eta}([u^\eta])$ with $\eta<\lambda$ and for every $\epsilon>0$ there exists a strong $(\epsilon,d_{FS})$-chain from $[u]$ to $[w_1]$ whose points belong to $M_{\eta,\lambda}\cup \{[w_1]\}$.
        
        If $[w_1]=[u^{\xi_1}]$ for some $\xi_1<\xi_0$, we repeat the same argument replacing $\xi_1$ with $\xi_0$ and observing that $M_{\xi_1,\lambda}$ is internally $\mathcal{SC}_{d_{FS}}$-transitive. 
        
        Since there is no infinitely decreasing sequence of ordinals, we can repeat the previous construction up to a certain  $k\in\mathbb{N}$, after which we must have $[w_{k+1}]\notin \{[u^\gamma]\}_{\gamma<\lambda}$.
        Then we set $[u^\lambda]:=[w_{k+1}]$ and property (b$_{\lambda+1}$) is satisfied.
    \end{itemize}
    Consider the sequence $\{[u^\gamma]\}_{\gamma\le \lambda}$. 
    By construction, property (c$_{\lambda+1}$) is verified, and from Lemma \ref{thm_existence} and Remark \ref{rem_strong} property (a$_{\lambda+1})$ follows.
    
    \item[2.] If $\lambda$ is a limit ordinal, since $|\{[u^\xi] \ |\ \gamma\le \xi <\lambda\}|\ge \aleph_0$ for every $\gamma<\lambda$, by compactness (see \cite[Th. 26.9, p. 169]{munkres2013topology}) we can pick 
    $$
    [v]\in \bigcap_{\gamma<\lambda}(\{[u^\xi] \ |\ \gamma\le \xi <\lambda\})'.
    $$
    Therefore, for every $[u]\in S_{\alpha_\eta}([u^\eta])$ with $\eta<\lambda$ and for every $\epsilon>0$ there exists a strong $(\epsilon,d_{FS})$-chain from $[u]$ to $[v]$ whose points belong to $M_{\eta,\lambda}\cup \{[v]\}$.
    Indeed, fix $\eta<\lambda$ and $\epsilon>0$ and take $[u]\in S_{\alpha_\eta}([u^\eta])$. 
    Then there exists $\eta<\xi<\lambda$ such that $d_{FS}([u^\xi],[v])< \epsilon/2$.
    By property (b$_\lambda$), there exists a strong $( \epsilon/2,d_{FS})$-chain $[u]=[y_0],[y_1],\ldots,[y_{n}]=[u^\xi]$ with $\{h_0,\ldots,h_{n-1}\}\subseteq\mathbb{N}$ from $[u]$ to $[u^\xi]$ whose points belong to $M_{\eta,\xi}\cup \{[u^\xi]\}$. 
    By the triangle inequality, it follows that $[u]=[y_0],[y_1],\ldots,[y_{n-1}],[v]$ with $h_0,\ldots,h_{n-1}$ is a strong $(\epsilon,d_{FS})$-chain from $[u]$ to $[v]$ whose points belong to $M_{\eta,\xi}\cup \{[v]\}\subseteq M_{\eta,\lambda}\cup \{[v]\}$.
    Indeed,
    \begin{equation*}
        \begin{aligned}
            \sum_{i=0}^{n-2}d_{\text{FS}}(T^{(h_i)}([y_i]),[y_{i+1}])&+ d_{\text{FS}}(T^{(h_{n-1})}([y_{n-1}]),[v])\le \\
            &\le \sum_{i=0}^{n-1}d_{\text{FS}}(T^{(h_i)}([y_i]),[y_{i+1}]) + d_{\text{FS}}([u^\xi],[v])< \frac \epsilon 2 + \frac \epsilon 2=\epsilon.
        \end{aligned}
    \end{equation*}
    
    Consider now the following two cases.
    \begin{itemize}
        \item[•] If $[v]\notin \{[u^\gamma]\}_{\gamma<\lambda}$, then we set $[u^\lambda]:=[v]$ and property (b$_{\lambda+1}$) is satisfied.
        
        \item[•] Assume that $[v]=[u^{\xi_0}]$ for some $\xi_0<\lambda$. 
        By (b$_\lambda$) and the property of $[v]$ shown above, it follows that $M_{\xi_0,\lambda}$ is internally $\mathcal{SC}_{d_{FS}}$-transitive.
        Since this set cannot be closed, there exists $[w_1]$ such that 
        $[w_1] \in \left(M_{\xi_0,\lambda}\right)'\setminus M_{\xi_0,\lambda}$, so that $[w_1]\notin \{[u^\gamma]\}_{\xi_0 \le \gamma<\lambda}$. 
        Since $M_{\xi_0,\lambda}$ is internally $\mathcal{SC}_{d_{FS}}$-transitive and $[w_1]\in (M_{\xi_0,\lambda})'$, by the triangle inequality we have that for every $[u]\in M_{\xi_0,\lambda}$ and for every $\epsilon>0$ there exists a strong $(\epsilon,d_{FS})$-chain from $[u]$ to $[w_1]$ whose points belong to $M_{\xi_0,\lambda}\cup \{[w_1]\}$.
        By property (b$_\lambda$) and using the transitivity of $\mathcal{SC}_{d_{FS}}$, we have that for every $[u]\in S_{\alpha_\eta}([u^\eta])$ with $\eta<\lambda$ and for every $\epsilon>0$ there exists a strong $(\epsilon,d_{FS})$-chain from $[u]$ to $[w_1]$ whose points belong to $M_{\eta,\lambda}\cup \{[w_1]\}$.
        
        If $[w_1]=[u^{\xi_1}]$ for some $\xi_1<\xi_0$, we repeat the same argument replacing $\xi_1$ with $\xi_0$ and observing that the set $M_{\xi_1,\lambda}$ is internally $\mathcal{SC}_{d_{FS}}$-transitive.      
        
        Since there is no infinitely decreasing sequence of ordinals, we can repeat the previous construction up to a certain  $k\in\mathbb{N}$, after which we must have $[w_{k+1}]\notin \{[u^\gamma]\}_{\gamma<\lambda}$.
        Then we set $[u^\lambda]:=[w_{k+1}]$ and property (b$_{\lambda+1}$) is satisfied. 
    \end{itemize}
By construction, the sequence $\{[u^\gamma]\}_{\gamma\le \lambda}$ verifies property (c$_{\lambda+1}$), and from Lemma \ref{thm_existence} and Remark \ref{rem_strong} property (a$_{\lambda+1})$ follows.
\end{itemize}
Assuming that $(\mathbb{P}(\mathcal{H}),\{T_{U,D}^{(k)}\}_{k\in\mathbb{N}})$ does not have a closed and internally $\mathcal{SC}_{d_{FS}}$-transitive subset, we find that the application
$$
\gamma\mapsto [u^\gamma] \quad (0\le\gamma<\lambda)
$$
is a bijection between $\lambda$ and the set $\{[u^\gamma]\}_{0\le\gamma<\lambda}$, so we have that $|\{[u^\gamma]\}_{0\le\gamma<\lambda}|=|\lambda|$. 
By Hartogs' Lemma (see \cite{hartogs1915problem}), we can take $\lambda$ so large that $|\{[u^\gamma]\}_{0\le \gamma<\lambda}|>|\mathbb{P}(\mathcal{H})|$, which is a contradiction. 

Therefore, there exists an ordinal $\nu<\lambda$ such that it is impossible to define the state $[u^\nu]$.
More precisely, if $\nu$ is a successor ordinal, one cannot define the state $[u^\nu]$ if one of the following cases holds:
\begin{itemize}
    \item The set $S_{\beta_{\nu-1},\alpha_{\nu-1}}([u^{\nu-1}])$ is closed, so that
    $(S_{\beta_{\nu-1},\alpha_{\nu-1}}([u^{\nu-1}]))'\subseteq S_{\beta_{\nu-1},\alpha_{\nu-1}}([u^{\nu-1}]).$ 
    This means that $S_{\beta_{\nu-1},\alpha_{\nu-1}}([u^{\nu-1}])$ is a closed, internally $\mathcal{SC}_{d_{FS}}$-transitive subset.
    
    \item The set $S_{\beta_{\nu-1},\alpha_{\nu-1}}([u^{\nu-1}])$ is not closed and there exists an ordinal $\xi<\nu$ such that $M_{\xi,\nu}$ is a closed, internally $\mathcal{SC}_{d_{FS}}$-transitive subset.
    \end{itemize}

On the other hand, if $\nu$ is a limit ordinal, then we cannot define the state $[u^\nu]$ if for every 
$$
[v]\in\bigcap_{\gamma<\nu}(\{[u^\xi] \ |\ \gamma\le \xi <\nu\})',
$$ 
we have that $[v]\in\{[u^\gamma]\}_{\gamma<\nu}$. 
This implies that there exists an ordinal $\xi<\nu$ such that $M_{\xi,\nu}$ is a closed, internally $\mathcal{SC}_{d_{FS}}$-transitive subset.
\end{proof}

\begin{remark}
    Note that the previous proof does not use the Axiom of Choice, as the operation of selecting a point from closed sets in a compact metric space can be performed by means of a selection function that exists in ZF+DC (\cite{Kura_}).
\end{remark}

\begin{corollary}\label{cor_basin}
Let $U$ be a linear operator, $D$ an admissible selector and consider the system $(\mathbb{P}(\mathcal{H}),\{T_{U,D}^{(k)}\}_{k\in\mathbb{N}})$. 
Then every state $[u]\in\mathbb{P}(\mathcal{H})$ is in strong chain relation with every state of a certain topologically closed, internally strong chain-transitive subset $S\subseteq \mathbb{P}(\mathcal{H})$.
\end{corollary}

\begin{proof}
    Take $[u^0]\in \mathbb{P}(\mathcal{H})$. 
    In the proof of Theorem \ref{th_rec_}, it has been established that there exists a least ordinal $\nu$ such that the set $\{[u^\gamma]\}_{\gamma<\nu}$ is defined, while it is impossible to define the state $[u^\nu]$.  
    If $\nu$ is a successor ordinal, then one of the following holds:
\begin{itemize}
    \item The set $S_{\beta_{\nu-1},\alpha_{\nu-1}}([u^{\nu-1}])$ is closed and internally $\mathcal{SC}_{d_{FS}}$-transitive.
    Moreover, by construction of the set $\{[u^\gamma]\}_{\gamma<\nu}$, we have $[u^0]\,\mathcal{SC}_{d_{FS}}\, S_{\beta_{\nu-1},\alpha_{\nu-1}}([u^{\nu-1}])$ and the claim is verified with $[u]=[u^0]$ and $S=S_{\beta_{\nu-1},\alpha_{\nu-1}}([u^{\nu-1}])$.
    
    \item The set $S_{\beta_{\nu-1},\alpha_{\nu-1}}([u^{\nu-1}])$ is not closed and there exists an ordinal $\xi<\nu$ such that the set $M_{\xi,\nu}$ is closed and internally $\mathcal{SC}_{d_{FS}}$-transitive. Again, by construction, we have $[u^0]\,\mathcal{SC}_{d_{FS}}\,M_{\xi,\nu}$ and the claim is verified with $[u]=[u^0]$ and $S=M_{\xi,\nu}$.
\end{itemize}
If $\nu$ is a limit ordinal, then there exists an ordinal $\xi<\nu$ such that $M_{\xi,\nu}$ is closed and internally $\mathcal{SC}_{d_{FS}}$-transitive and $[u^0]\,\mathcal{SC}_{d_{FS}} \,M_{\xi,\nu}$.
Therefore, the claim is verified with $[u]=[u^0]$ and $S=M_{\xi,\nu}$.
\end{proof}

\section{Conclusions}\label{sec_concl}
Our work establishes a no-go for operational irreversibility along arbitrary realizations of the collapse dynamics for finite dimensional open quantum systems. Specifically, Theorems \ref{th_rec_} and \ref{th_rec_2} tell us the following fact (which was indicated in the Introduction as our main result): in every finite-dimensional, open quantum system, there is a topologically closed subset $S$ where any two states are connected via pseudo-orbits entirely made of points of $S$ with arbitrarily small Fubini-Study error and energetic cost. By Def. \ref{arrow}, this means that $S$ is operationally reversible: no operational arrow of time arises. 

It is perhaps useful to frame this result by analogy with the logical structure of some classical thermodynamics concepts. The second law provides a universal lower bound on irreversibility, while Carnot’s theorem is an existence result showing that this bound can in fact be saturated: it identifies ideal cycles that achieve reversibility under suitable structural constraints.
Landauer’s principle plays a similar role to the second law, including information in the entropy balance and thus giving a universal lower bound on dissipation for logically irreversible operations. Our result stands in a similar relation to Landauer’s principle as Carnot’s theorem does to the second law of thermodynamics: it is an existence statement identifying conditions under which the bound is asymptotically saturated. 

Let us finally remark that our argument crucially relies on the retention of the \textit{infinite} record of outcomes provided by the selector. This infinite history is essential because, to steer the system toward a target state with arbitrarily small FS error $\epsilon>0$, the perturbations must be tailored to the specific sequence of realized collapse outcomes encountered along the chain. Although only finitely many corrections are required to achieve any fixed $\epsilon$-precision, the number of these corrections becomes unbounded, in general, as $\epsilon\to 0$, which means that, in principle, one needs access to the entire forward-itinerary record to define the appropriate steering strategy at every scale. A similar observation is key in Bennett's resolution of Maxwell's demon paradox, where he showed that a demon's finite memory forces logically irreversible erasure to accommodate new data, thereby incurring an entropic cost that preserves the second law \cite{bennett1982thermodynamics}. Our no-erasure framework mirrors the idealized demon scenario, where an infinite memory sidesteps dissipation. Yet, as Bennett argued for any physical demon, real-world implementations face finite-size constraints that inevitably lead to erasure. Analogously, any practical realization of our procedure with a truncated record would introduce operational irreversibility, which recovers the second law.

\section*{Declarations}
The authors declare that they have no conflicts of interest.

The manuscript has no associated data.

No funding was recieved for this research. 

\bibliographystyle{spmpsci}
\bibliography{quantum_bib}

\end{document}